\documentclass[a4papers,runningheads]{llncs}

\usepackage[firstpage]{draftwatermark}
\newif\ifblind
\blindfalse

\newif\ifdraft
\draftfalse

\newif\iflong
\longfalse

\usepackage{microtype}

\usepackage[T1]{fontenc}
\usepackage{newtxtext,newtxmath}
\usepackage[scaled=0.81]{beramono}

\usepackage{graphicx}

\usepackage{hyperref}

\usepackage[inline]{enumitem}
\setlist[enumerate]{label=\emph{\roman*})}

\usepackage{tikz}
\usetikzlibrary{shapes,arrows,arrows.meta,positioning,calc,fit,intersections,through,petri,patterns,patterns.meta}
\tikzset{
  node distance=9mm and 5mm,
  acq/.style={draw=none,preaction={fill=acquirecol},pattern={Lines[distance=2pt,line width=0.5pt,angle=90]}},
  rel/.style={draw=none,preaction={fill=releasecol},pattern={Lines[distance=2pt,line width=0.5pt,angle=0]}},
}

\tikzset{
    >=latex,arr l2 t1/.style={
        append after command={
            (\tikzlastnode.north west) edge[#1, thick, l2t1col] ++(135:5mm)
        }
    },
    arr l2 t2/.style={
        append after command={
            ($(\tikzlastnode.north west)!0.25!(\tikzlastnode.north east)$) edge[#1, thick, l2t2col] ++(90:4mm)
        }
    },
    arr l1 t2/.style={
        append after command={
            ($(\tikzlastnode.north west)!0.75!(\tikzlastnode.north east)$) edge[#1, thick, l1t2col] ++(90:4mm)
        }
    },
    arr l1 t1/.style={
        append after command={
            (\tikzlastnode.north east) edge[#1, thick, l1t1col] ++(45:5mm)
        }
    },
node arrows/.style={#1}
  }

\DeclareRobustCommand{\inlineplace}[1]{\tikz[baseline=-0.5ex] \node[circle, fill=#1, inner sep=0pt, minimum size=2ex] {};}

\newlist{features}{description}{1}
\iflong
\setlist[features]{font=\bfseries,        labelsep=0.6em,        leftmargin=*,          topsep=\parskip,       itemsep=0.5\parskip,   listparindent=0mm,     itemindent=-2.5mm      }
\else
\setlist[features]{font=\bfseries,        labelsep=0.6em,        leftmargin=*,          topsep=\parskip,       itemsep=0.5\parskip,   listparindent=0mm,     itemindent=-1.5mm      }
\fi

\usepackage{ragged2e}

\usepackage{booktabs}
\usepackage{relsize}
\usepackage{changepage}            \usepackage[skip=5pt,font=small,strut=off]{caption}
\usepackage[skip=1pt]{subcaption}  

\DeclareGraphicsExtensions{.pdf,.png}

\usepackage{amsmath,amsfonts}
\usepackage{mathtools}
\usepackage{array}
\usepackage{xspace}
\usepackage{multicol,multirow}
\usepackage{xurl}
\usepackage{fontawesome}
\usepackage{stmaryrd}
\usepackage{mathpartir}
\usepackage{dsfont}

\definecolor{jimplecol}{HTML}{1b9e77}
\definecolor{rbcol}{HTML}{d95f02}
\colorlet{petrifycol}{rbcol}
\definecolor{pncol}{HTML}{7570b3}
\definecolor{othercol}{HTML}{e7298a}
\colorlet{graycol}{black!25}

\definecolor{l1t1col}{HTML}{ca0020}
\definecolor{l1t2col}{HTML}{f4a582}
\definecolor{l2t2col}{HTML}{92c5de}
\definecolor{l2t1col}{HTML}{0571b0}

\definecolor{acquirecol}{HTML}{a6cee3}
\definecolor{releasecol}{HTML}{b2df8a}
\DeclareRobustCommand{\stylenode}[1]{\begin{tikzpicture}[baseline=-0.5ex]\draw[#1] (0,0) rectangle (3mm,3mm);\end{tikzpicture}}

\newcommand{\running}{\text{\faCaretRight}}
\newcommand{\terminated}{\text{\faClose}}
\newcommand{\spawnMark}{\text{\faPlayCircleO}}
\newcommand{\termMark}{\text{\faStopCircleO}}
\newcommand{\naturals}{\ensuremath{\mathds{N}}}
\newcommand{\encoding}{\ensuremath{\mathcal{E}}}
\newcommand{\translation}{\ensuremath{\mathcal{T}}}
\newcommand{\joins}[1]{\ensuremath{\bowtie\!#1}}
\newcommand{\rets}[1]{\ensuremath{\hookrightarrow\!#1}}
\newcommand{\pndeadlock}{\ensuremath{\mathds{D}}}

\newcommand{\nextloc}{\ensuremath{\curvearrowright}}

\usepackage{boxedminipage}

\usepackage{listings}

\lstdefinestyle{displayed}{
  numbers=left, firstnumber=last, 
  frame=single, breaklines=true,
  numbersep=2pt,
frame=tb,
numberstyle=\tiny,
  tabsize=2,
  captionpos=b,
xleftmargin=0mm, xrightmargin=0mm, basicstyle=\ttfamily\scriptsize,
  keywordstyle=\bfseries\ttfamily,
  commentstyle=\color{darkgray}\itshape\ttfamily,
  keepspaces=true,
  columns=fixed,
  escapeinside={(*}{*)},
  mathescape=true,
  showstringspaces=false
}

\lstdefinestyle{displayed-nln}{
  style=displayed,
  numbers=none,
  basicstyle=\ttfamily\footnotesize,
}

\lstdefinestyle{plain}{
  numbers=none, frame=none, breaklines=false,
tabsize=2,
xleftmargin=2mm, xrightmargin=2mm, basicstyle=\ttfamily\scriptsize,
  keywordstyle=\bfseries\ttfamily,
  commentstyle=\color{darkgray}\itshape\ttfamily,
  keepspaces=true,
  columns=fullflexible,
  escapeinside={(*}{*)},
  mathescape=true,
  showstringspaces=false
}

\lstdefinestyle{inlined} {
  numbers=none,
  frame=none
}

\lstdefinelanguage{JavaRecent}[]{Java}
{
  morekeywords={var,yield}
}

\lstdefinelanguage{Jimple}[]{}
{
  morekeywords=[2]{method,monitor_enter,monitor_exit,noop,goto,if,then,switch,return,invoke,invokestatic,throw},
  morekeywords=[4]{label,error,ok},
  literate={:=}{$:$=}2
  {!}{$\lnot\ $}1
  {!=}{$\neq$}1
  {||}{$\lor\;$}2
  {&&}{$\land\;$}2
  {==>}{$\,\Longrightarrow\;$}4
  {forall}{$\forall$}1
  {exists}{$\exists$}1
  {:}{$\colon$}1
  {::}{$\colon\!\!\colon$}1
  {:=}{$:$=}2
  {<}{$<$}1
  {<=}{$\le$}1
  {>}{$>$}1
  {>=}{$\ge$}1
  ,
  mathescape=true,
  escapeinside={(*}{*)},
  identifierstyle=\ttfamily,
  keywordstyle=[2]{\color{jimplecol}\bfseries\ttfamily},
  keywordstyle=[3]{\color{jimplecol}\bfseries\ttfamily},
  keywordstyle=[4]{\itshape\underbar},
}

\lstdefinelanguage{RockBottom}[]{}
{
  morekeywords=[2]{def,begin,end,acquire,release,skip,goto,jump,call,read,write,fork,join},
  morekeywords=[4]{entry,exit,label,error,ok},
  literate={:=}{$:$=}2
  {!}{$\lnot\ $}1
  {!=}{$\neq$}1
  {||}{$\lor\;$}2
  {&&}{$\land\;$}2
  {==>}{$\,\Longrightarrow\;$}4
  {forall}{$\forall$}1
  {exists}{$\exists$}1
  {:}{$\colon$}1
  {::}{$\colon\!\!\colon$}1
  {:=}{$:$=}2
  {<}{$<$}1
  {<=}{$\le$}1
  {>}{$>$}1
  {>=}{$\ge$}1
  ,
  mathescape=true,
  escapeinside={(*}{*)},
  identifierstyle=\ttfamily,
  keywordstyle=[2]{\color{rbcol}\bfseries\ttfamily},
  keywordstyle=[3]{\color{rbcol}\bfseries\ttfamily},
  keywordstyle=[4]{\itshape\underbar},
}

\newcommand{\J}[1]{\mbox{\lstinline[basicstyle=\ttfamily,language=JavaRecent]|#1|}}
\newcommand{\Bc}[1]{\mbox{\lstinline[basicstyle=\ttfamily,language=JVMIS]|#1|}}
\newcommand{\Ji}[1]{\mbox{\lstinline[basicstyle=\ttfamily,language=Jimple]|#1|}}
\newcommand{\RB}[1]{\mbox{\lstinline[basicstyle=\ttfamily,language=RockBottom]|#1|}}
\newcommand{\nt}[1]{\ensuremath{\textsl{#1}}}
\newcommand{\locs}{\ensuremath{\mathcal{L}}}

\newcommand{\petrify}{{\smaller[0.5]{\textsc{Petri\-fy}}}\xspace}
\newcommand{\tech}{\petrify}
\newcommand{\jpetrify}{{\smaller[0.5]{j\textsc{Petri\-fy}}}\xspace}
\newcommand{\tool}{\jpetrify}
\newcommand{\rb}{{\smaller[0.5]{\textsc{rb}}}\xspace}

\newcommand{\success}{\textcolor{green}{\faCheck}}
\newcommand{\falsepos}{\textcolor{red}{\faPlus}}
\newcommand{\falseneg}{\textcolor{red}{\faMinus}}
\newcommand{\unsupp}{\textcolor{olive}{\faBan}}
\newcommand{\timeout}{\textcolor{orange}{\faHourglassEnd}}
\newcommand{\isCorrect}{\textcolor{green}{\faCheckSquare}}
\newcommand{\isBug}{\textcolor{red}{\faBug}}

\newcounter{examplecnt}

\newenvironment{expara}[1][]{\vspace{1pt}\noindent\refstepcounter{examplecnt}\textbf{Example~\theexamplecnt\ifthenelse{\equal{#1}{}}{}{(#1)}. }}{\hfill\scriptsize$\blacksquare$\vspace{1pt}\noindent\ignorespaces}

\let\llncssubparagraph\subparagraph
\let\subparagraph\paragraph
\usepackage[compact]{titlesec}
\let\subparagraph\llncssubparagraph

\makeatletter
\newcommand\nicepar{\@startsection{paragraph}{4}{\z@}{-1\p@ \@plus -1\p@ \@minus -1\p@}{-0.3em \@plus -0.1em \@minus -0.1em}{\itshape\normalsize}}
\newcommand\nicesection{\@startsection{section}{1}{\z@}{-8\p@ \@plus -3\p@ \@minus -3\p@}
  {4\p@ \@plus 2\p@ \@minus 1\p@}
  {\normalfont\large\bfseries\boldmath}}
\newcommand\nicesubsection{\@startsection{subsection}{2}{\z@}{-5\p@ \@plus -2\p@ \@minus -2\p@}{2\p@ \@plus 1\p@ \@minus 1\p@}{\normalfont\normalsize\bfseries\boldmath}}
\newcommand\nicesubsubsection{\@startsection{subsubsection}{3}{\z@}{-3\p@ \@plus -1\p@ \@minus -1\p@}{-0.3em \@plus -0.1em \@minus -0.1em}{\normalfont\normalsize\bfseries\boldmath}}
\makeatother
\renewcommand{\subsubsection}{\nicesubsubsection}
\renewcommand{\subsection}{\nicesubsection}
\renewcommand{\section}{\nicesection}
\renewcommand{\paragraph}{\nicepar}

\usepackage{ifthen}

\ifdraft
\usepackage[colorinlistoftodos,textsize=scriptsize,obeyFinal,textwidth=33mm]{todonotes}

\DeclareDocumentCommand{\ReviewNote}{s o m O{white}}{\todo[color=#4,\IfBooleanTF{#1}{inline}{}]{\IfNoValueF{#2}{\textbf{#2:}\xspace}#3}
}
\else
\DeclareDocumentCommand{\ReviewNote}{s o m O{white}}{}
\fi

\DeclareDocumentCommand{\caf}{s m}{\IfBooleanTF{#1}{\ReviewNote*{#2}[yellow]}{\ReviewNote{#2}[yellow]}}
\DeclareDocumentCommand{\aks}{s m}{\IfBooleanTF{#1}{\ReviewNote*{#2}[red!70!white]}{\ReviewNote{#2}[red!70!white]}}

\DeclareDocumentCommand{\pntrans}{s m m O{} O{} O{}}{\IfBooleanTF{#1}{\node[transition,minimum size=3mm,label={right:\scriptsize #5},label=left:{\scriptsize #3\ #4},#6] (#2) {};}{\node[transition,minimum size=3mm,label={below:\scriptsize #5},label=above:{\scriptsize #3\ #4},#6] (#2) {};}
}

\DeclareDocumentCommand{\pnplace}{s m m O{} O{} O{}}{\IfBooleanTF{#1}{\node[place,minimum size=3mm,label={right:\scriptsize #5},label=left:{\scriptsize #3\ #4},#6] (#2) {};}{\node[place,minimum size=3mm,label={below:\scriptsize #5},label=above:{\scriptsize #3\ #4},#6] (#2) {};}
}

\DeclareRobustCommand{\MTLoperator}[4] 
{\ensuremath{\ifthenelse{\not \equal{#2}{}} {{#1}_{{#2}}^{#3}} {{#1}}
	 \ifthenelse{\not \equal{#4}{}} {\!\left({#4}\right)} {}
  }}

\DeclareDocumentCommand{\MTLop}{m o o o}{\ensuremath{\IfNoValueTF{#3}{#1}{{#1}_{{#3}}^{{#4}}}\IfNoValueTF{#2}{}{\!\left(#2\right)}\xspace }}

\DeclareDocumentCommand{\Gtl}{o}{\MTLop{\mathsf{G}}[#1]}
\DeclareDocumentCommand{\Ftl}{o}{\MTLop{\mathsf{F}}[#1]}
\newcommand{\Atl}{\MTLop{\mathsf{A}}}
\newcommand{\Etl}{\MTLop{\mathsf{E}}}

\begin{document}
\SetWatermarkText{\hspace*{3.7in}\raisebox{8.55in}{\mbox{\includegraphics[width=1.4cm]{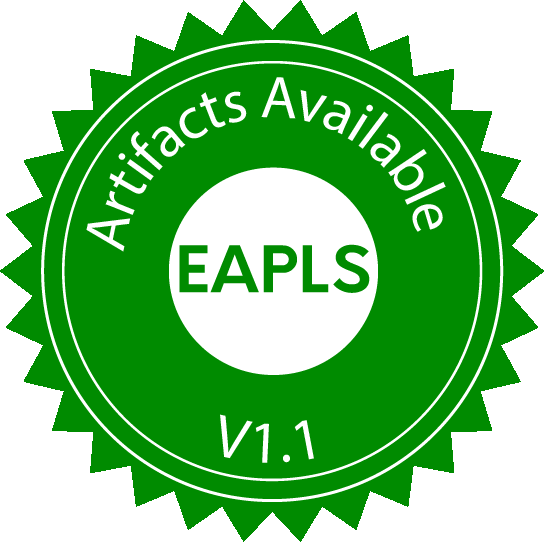}\includegraphics[width=1.4cm]{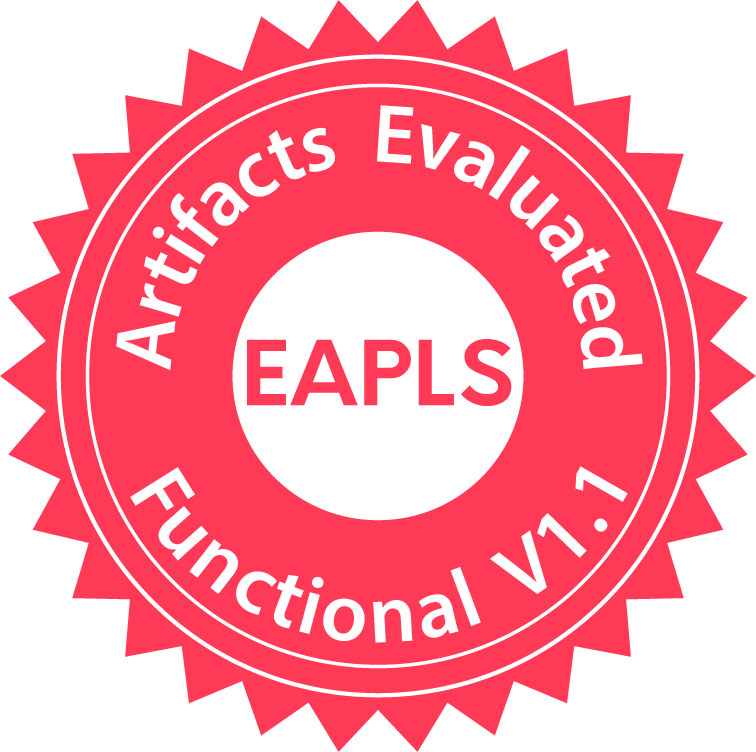}}}}
\SetWatermarkAngle{0}

\title{\petrify: Petri-net Model-Checking of \\Concurrency Properties in Java Bytecode}
\title{\petrify: Petri-net Based Analysis of \\Concurrency Properties in Java Bytecode}
\titlerunning{Petri-based Analysis of Bytecode Concurrency}

\ifblind
  \author{Anonymous authors}
\else
\author{Akshatha Shenoy\inst{1}\orcidID{0009-0004-2439-1656} \and
Carlo A.\ Furia \inst{1}\orcidID{0000-0003-1040-3201}}
\institute{Software Institute, USI Università della Svizzera italiana, Lugano, Switzerland \\
  \email{shenoa@usi.ch} $\quad$ \url{bugcounting.net}
}
\fi

\maketitle

\begin{abstract}
  The landscape of automated formal verification 
  is populated by techniques that make prominently different trade-offs:
  some focus on expressiveness and precision, supporting the verification
  of complex properties;
  others favor scalability and practicality, so that they are applicable
  to larger programs using different features.
  This paper presents \tech, a novel automated verification technique
  for concurrency properties that achieves a distinctive trade-off.
  \tech encodes the semantics of Java bytecode programs
  into Petri nets (PNs), which can be analyzed by state-of-the-art
  model checking tools such as LoLA.
  As our experiments demonstrate, 
  \tech's approach offers an interesting combination of expressiveness
  and practicality:
  PNs are a fairly precise encoding of the concurrent behavior of programs;
  at the same time, \tech's PN encoding is succinct,
  so that its analysis remains quite insensitive to parameter size.
  Another practical benefit of targeting bytecode is that
  \tool, the prototype tool that implements the \tech technique,
  is applicable, in principle, to programs written in any version of Java and even a subset of Kotlin (another language that compiles to Java bytecode) while other similar tools are limited to older versions of Java.
  While this paper's experiments focus on analyzing fundamental 
  properties like deadlock,
  \tech's approach lends itself to be extended to other
  kinds of concurrency analysis, which we plan to tackle in future work.
\end{abstract}
\counterwithout{lstlisting}{chapter}
\section{Introduction}

Every automated formal program analysis technique
has to contend with the expressiveness vs. scalability trade-off.
Techniques that favor expressiveness
support precisely verifying all sorts of user-defined properties,
but may struggle to analyze large programs.
In contrast,
techniques that target scalability are based on bespoke abstractions
that approximate only specific hardcoded properties,
but are applicable to
realistic-size programs.
In the domain of concurrent verification, for example,
software model checkers (e.g., Java Pathfinder~\cite{JPF})
favor expressiveness as they support temporal logic specifications
and other kinds of program annotations,
and can analyze path conditions precisely.
In contrast,
custom static analyzers that focus on a single property
(e.g., data races for Infer’s RacerD~\cite{RacerD}
or deadlocks for JaDA~\cite{JaDA})
are usually more efficient but cannot verify other kinds of properties
and may suffer from imprecision.
Another important dimension of practicality for any formal verification tool
is \emph{language support}.
Most analysis techniques work at the source-code level, which entails
that they may struggle to keep up with the evolution of modern languages.
Consider, again, the example of Java: tools like Pathfinder and JaDA
were developed for earlier versions of Java, and hence cannot analyze programs that
include recently introduced features---even if the features themselves do not affect the concurrent behavior.
All the more so,
applying a verification technique to work on a (subset of a) different language (even one that is similar to Java) is usually onerous---in terms of both adapting the technique and developing a suitable implementation toolchain.

This paper presents \tech,
a program analysis technique for concurrent Java programs
that explores a novel trade-off between expressiveness, scalability, and language support.
\tech relies on two key ideas.
First,
it encodes the concurrent behavior of a Java program as a Petri net.
Thanks to recent advances in their algorithmic verification,
PNs have become an appealing abstract model of concurrent computation,
which combines a high expressiveness with highly optimized model-checking tools. Second, \tech
translates directly from Java bytecode,
rather than working at the source code level.
This makes it a technique that is not tied too closely
to a specific version of Java,
and even works for programs written in (an imperative subset of) Kotlin---another language that compiles to Java bytecode---which provides an additional element of flexibility and practicality.

In our experiments to demonstrate \tech's capabilities,
we implemented it in a tool called \tool,
and applied it to detect deadlocks
in 36 Java and 3 Kotlin programs taken or adapted from various benchmarks.
Although our prototype implementation of \tool
currently only supports a limited number of properties out of the box,
the experiments demonstrated some of its strengths,
which complement
other state-of-the-art tools such as Java Pathfinder and JaDA.
In particular,
\tool is applicable to
Java programs written in any version of the language,
whereas other tools usually cannot process versions more recent than Java~11.
In addition, \tech's abstractions are insensitive
to the parameter size, so that \tool scales up to examples
that heavier, more precise tools cannot handle.

\paragraph{Contributions.}
The paper makes the following main contributions:
\begin{enumerate*}
\item \tech: a novel technique to analyze concurrency properties
  of bytecode programs based on a flow- and context-sensitive, path-insensitive
  encoding of bytecode programs into Petri nets.
\item \tool: an implementation of \tech based on the Soot static analyzer
  and the LoLA Petri net model checker.
\item An experimental evaluation of \tool on 39 programs.
  \ifblind\else
  \item The implementation of \tool and the experimental artifacts are available~\cite{replication}.
  \fi
\end{enumerate*}
For lack of space, some technical details and examples
have been moved to the appendix. The main text
remains self-contained, while focusing on key high-level details.

\begin{figure}[!tb]
  \begin{subfigure}{1.0\linewidth}
\begin{lstlisting}[language=JavaRecent,style=displayed-nln,numbers=left]
void main() { (*\label{l:td:main}*)
    Object lock1 = new Object();
    Object lock2 = new Object();
    Runnable r_12 = () -> { synchronized (lock1) { synchronized (lock2) {}}};
    Runnable r_21 = () -> { synchronized (lock2) { synchronized (lock1) {}}};
    Thread t_12 = new Thread(r_12);  (*\label{l:td:write:thread_12}*)
    Thread t_21 = new Thread(r_21);  (*\label{l:td:write:thread_21}*)
    t_12.start(); (*\label{l:td:fork:thread_12}*)
    t_21.start(); (*\label{l:td:fork:thread_21}*)
}
\end{lstlisting}
  \caption{A two-threaded Java program that may deadlock.}
  \label{fig:ex:java}
  \end{subfigure}
  \\
  \begin{subfigure}{1.0\linewidth}
  \begin{subfigure}{0.35\linewidth}
    \scriptsize
    \begin{tabular}{rl}
                             & \RB{def main} \\
      $\RB{entry}_{\RB{main}}\colon$ & \RB{begin} \\
      $\ell_0\colon$ & \RB{write lock1} \\
      $\ell_1\colon$ & \RB{write lock2} \\
      $\ell_2\colon$ & \RB{write t_12} \\
      $\ell_3\colon$ & \RB{write t_21} \\
      $\ell_4\colon$ & \RB{fork r_12 t_12} \\
      $\ell_5\colon$ & \RB{fork r_21 t_21} \\
      $\RB{exit}_{\RB{main}}\colon$ & \RB{end}
    \end{tabular}
  \end{subfigure}
  \begin{subfigure}{0.32\linewidth}
    \scriptsize
    \begin{tabular}{rl}
                             & \RB{def r_12} \\
      $\RB{entry}_{\RB{r_12}}\colon$ & \RB{begin} \\
      $\ell_6\colon$ & \RB{acquire lock1} \\
      $\ell_7\colon$ & \RB{acquire lock2} \\
      $\ell_8\colon$ & \RB{release lock2} \\
      $\ell_9\colon$ & \RB{release lock1} \\
      $\RB{exit}_{\RB{r_12}}\colon$ & \RB{end}
    \end{tabular}
  \end{subfigure}
  \begin{subfigure}{0.32\linewidth}
    \scriptsize
    \begin{tabular}{rl}
                             & \RB{def r_21} \\
      $\RB{entry}_{\RB{r_21}}\colon$ & \RB{begin} \\
      $\ell_{10}\colon$ & \RB{acquire lock2} \\
      $\ell_{11}\colon$ & \RB{acquire lock1} \\
      $\ell_{12}\colon$ & \RB{release lock1} \\
      $\ell_{13}\colon$ & \RB{release lock2} \\
      $\RB{exit}_{\RB{r_21}}\colon$ & \RB{end}
    \end{tabular}
  \end{subfigure}
  \caption{\tech's translation of \autoref{fig:ex:java}'s program into the \rb intermediate language.}
  \label{fig:ex:rb}
\end{subfigure}  
\\[2mm]
\begin{subfigure}{0.94\linewidth}
  \centering
\begin{tikzpicture}[-latex]
  \matrix (main) [row sep=6mm,column sep=1.8mm,nodes={anchor=north}] {
    \pnplace{p-m-e}{\RB{m}}[$t_0$][\RB{entry}]
    &
    \pntrans{t-m-e}{\RB{m}}[$t_0$][\RB{entry}]
    &
    \pnplace{p-m-0}{}[][$\ell_0$]
    &
    \pntrans{t-m-0}{}[][$\ell_0$]
    &
    \pnplace{p-m-1}{}[][$\ell_1$]
    &
    \pntrans{t-m-1}{}[][$\ell_1$]
    &
    \pnplace{p-m-2}{}[][$\ell_2$]
    &
    \pntrans{t-m-2}{}[][$\ell_2$]
    &
    \pnplace{p-m-3}{}[][$\ell_3$]
    &
    \pntrans{t-m-3}{}[][$\ell_3$]
    &
    \pnplace{p-m-4}{}[][$\ell_4$]
    &
    \pntrans{t-m-4}{}[][$\ell_4$]
    &
    \pnplace{p-m-5}{}[][$\ell_5$]
    &
    \pntrans{t-m-5}{}[][$\ell_5$]
    &
    \pnplace{p-m-x}{\RB{m}}[$t_0$][\RB{exit}][]
    &
    \pntrans{t-m-x}{\RB{m}}[$t_0$][\RB{exit}]
    \\
  };
  \matrix (l12) [row sep=6mm,column sep=1.8mm,above=3mm of main.center,nodes={anchor=south}] {
    \pntrans{t-12-x}{\RB{t_12}}[][\RB{exit}]
    &
    \pnplace{p-12-x}{\RB{t_12}}[][\RB{exit}]
    &
    \pntrans{t-12-0-rel}{}[][$\ell_9$][node arrows={
      arr l1 t1={<->}, arr l1 t2={->}
      }]
    &
    \pnplace{p-12-0-rel}{}[][$\ell_9$]
    &
    \pntrans{t-12-1-rel}{}[][$\ell_8$][node arrows={
      arr l2 t1={<->}, arr l2 t2={->}
      }]
    &
    \pnplace{p-12-1-rel}{}[][$\ell_8$]
    &
    \pntrans{t-12-1}{}[][$\ell_7$][node arrows={
      arr l2 t1={<->}, arr l2 t2={<-}
      }]
    &
    \pnplace{p-12-1}{}[][$\ell_7$]
    &
    \pntrans{t-12-0}{}[][$\ell_6$][node arrows={
      arr l1 t1={<->}, arr l1 t2={<-}
      }]
    &
    \pnplace{p-12-0}{}[][$\ell_6$]
    &
    \pntrans{t-12-e}{\RB{t_12}}[][\RB{entry}]
    &
    \pnplace{p-12-e}{\RB{t_12}}[][\RB{entry}]
    \\
  };
  \matrix (l21) [row sep=6mm,column sep=1.8mm,below=4mm of main.center,nodes={anchor=south}] {
    \pntrans{t-21-x}{\RB{t_21}}[][\RB{exit}]
    &
    \pnplace{p-21-x}{\RB{t_21}}[][\RB{exit}]
    &
    \pntrans{t-21-0-rel}{}[][$\ell_{13}$][node arrows={
      arr l2 t1={->}, arr l2 t2={<->}
      }]
    &
    \pnplace{p-21-0-rel}{}[][$\ell_{13}$]
    &
    \pntrans{t-21-1-rel}{}[][$\ell_{12}$][node arrows={
      arr l1 t1={->}, arr l1 t2={<->}
      }]
    &
    \pnplace{p-21-1-rel}{}[][$\ell_{12}$]
    &
    \pntrans{t-21-1}{}[][$\ell_{11}$][node arrows={
      arr l1 t1={<-}, arr l1 t2={<->}
      }]
    &
    \pnplace{p-21-1}{}[][$\ell_{11}$]
    &
    \pntrans{t-21-0}{}[][$\ell_{10}$][node arrows={
      arr l2 t1={<-}, arr l2 t2={<->}
      }]
    &
    \pnplace{p-21-0}{}[][$\ell_{10}$]
    &
    \pntrans{t-21-e}{\RB{t_21}}[][\RB{entry}]
    &
    \pnplace{p-21-e}{\RB{t_21}}[][\RB{entry}]
    \\
  };

  \foreach \x/\y in {p-m-e/t-m-e,t-m-e/p-m-0,p-m-0/t-m-0,t-m-0/p-m-1,p-m-1/t-m-1,t-m-1/p-m-2,p-m-2/t-m-2,t-m-2/p-m-3,p-m-3/t-m-3,t-m-3/p-m-4,p-m-4/t-m-4,t-m-4/p-m-5,p-m-5/t-m-5,t-m-5/p-m-x,p-m-x/t-m-x,p-12-e/t-12-e,t-12-e/p-12-0,p-12-0/t-12-0,t-12-0/p-12-1,p-12-1/t-12-1,t-12-1/p-12-1-rel,p-12-1-rel/t-12-1-rel,t-12-1-rel/p-12-0-rel,p-12-0-rel/t-12-0-rel,t-12-0-rel/p-12-x,p-12-x/t-12-x,p-21-e/t-21-e,t-21-e/p-21-0,p-21-0/t-21-0,t-21-0/p-21-1,p-21-1/t-21-1,t-21-1/p-21-1-rel,p-21-1-rel/t-21-1-rel,t-21-1-rel/p-21-0-rel,p-21-0-rel/t-21-0-rel,t-21-0-rel/p-21-x,p-21-x/t-21-x} {\draw (\x) to (\y);}
  \draw (t-m-4) to (p-12-e);
  \draw (t-m-5) to (p-21-e);
  \pnplace{l-12-1}{}[\RB{t_12}][\RB{lock1}][above right=3mm and 5mm of t-m-x,tokens=1,fill=l1t1col]
  \pnplace{l-21-1}{}[\RB{t_21}][\RB{lock1}][below right=3mm and 5mm of t-m-x,tokens=1,fill=l1t2col]
  \pnplace{l-12-2}{}[\RB{t_12}][\RB{lock2}][above left=3mm and 5mm of p-m-e,tokens=1,fill=l2t1col]
  \pnplace{l-21-2}{}[\RB{t_21}][\RB{lock2}][below left=3mm and 5mm of p-m-e,tokens=1,fill=l2t2col]
  \node at (p-m-e) {$\bullet$};
  
\end{tikzpicture}
\caption{\tech's encoding of \autoref{fig:ex:rb}'s \rb program as a Petri net.
  To reduce clutter, the arrows connecting transitions $\ell_{6\text{--}9}$
  and $\ell_{10\text{--}13}$
  to the four places \inlineplace{l2t1col}, \inlineplace{l2t2col}, \inlineplace{l1t1col}, \inlineplace{l1t2col} are shortened and rely on colors
  to indicate the connected place.
}
  \label{fig:ex:pn}
  \end{subfigure}
  \caption{A simple concurrent Java program,
    and \tech's encoding in \rb and as a Petri net.}
  \label{fig:motivating-example}
\end{figure}

\section{An Overview of \tech}
\label{sec:example}

\autoref{fig:ex:java} shows a simple concurrent Java program with two threads
\J{thread_12} and \J{thread_21} that try to acquire a lock on variables \J{lock1} and \J{lock2}.
It is easy to see that the program will deadlock
if the two threads \emph{interleave} their lock acquisition operations,
leading to a state where \J{t_12} has a lock on \J{lock1},
\J{t_21} has a lock on \J{lock2},
and each thread waits for the other thread to release the lock it is holding.
Let's describe how \tech analyzes this program to find the deadlock scenario.

\tech works on \emph{bytecode} (produced by the Java compiler),
rather than on source Java code.
One advantage of targeting bytecode is that we can support 
any Java version, since new language features
(e.g., \autoref{fig:ex:java}'s instance main method and unnamed class,
which were introduced in Java~21 and still are preview features)
are
desugared into a stable set of bytecode instructions by the compiler.

\tech first translates the input bytecode program into the \rb program
shown in \autoref{fig:ex:rb}.
\rb is an intermediate language that is flow- and context-sensitive
but path-insensitive;
hence, it represents an approximation\iflong\footnote{
  As we explain in \autoref{sec:limitations},
  \tech's \rb translation is an overapproximation under certain conditions:
  the input Java program does not use features that are currently not supported,
  and the results of an alias analysis are sufficiently precise.
}
\fi
of the bytecode program's
executions.
We introduced \rb to
reduce the semantic gap between bytecode and Petri nets,
which simplifies the design of the overall translation 
and also reasoning about its correctness.
In this example, \rb captures all aspects
of concurrent behavior without information loss.

Then, \tech encodes the semantics of \autoref{fig:ex:rb}'s \rb program into the Petri net (PN) in \autoref{fig:ex:pn}.
\tech supports the PN format used by state-of-the-art analyzers such
as LoLA~\cite{LoLA}, so that its output can be fed to these tools to analyze
any properties of interest.
To check the presence of \emph{deadlocks},
\tech augments the output PN with additional ``monitoring'' places and transitions
(now shown in \autoref{fig:ex:pn} for simplicity)
and a suitable temporal logic formula.
With this input, LoLA quickly finds a PN execution that corresponds to \autoref{fig:ex:java}'s deadlock.

As you can glean from \autoref{fig:ex:pn},
the structure of the PN is clearly modular.
The central chain of nodes
(circular ``places'' and square ``transition'' in PN terminology)
encodes the control flow of the \J{main} function,
whereas the top and bottom chains correspond to the anonymous methods (lambdas) executed by threads \J{t_12} and \J{t_21}
respectively.
The arrows connecting transitions $\ell_4$ and $\ell_5$
to the entry places of \J{t_12} and \J{t_21}
represent the starting points of the threads parallel to \J{main}.
The modular structure of \tech's PN encoding
also helps map back an error trace given by the PN analyzer
to an execution of the original Java program.

\section{Related Work}
\label{sec:related-work}

There is a vast amount of research on
detecting \emph{concurrency} programming errors.
For space constraints,
we focus on techniques and tools that are applicable to Java
and are currently available;
as for any kind of program analysis,
they can be broadly classified in
static or dynamic, and according to their expressiveness vs.\ scalability.

\nicepar{Static techniques\,}
are usually sound but imprecise.
\emph{Deductive} verification offers high expressiveness,
as it can analyze user-defined complex properties;
however, it requires a significant amount of human effort, as
programs must be annotated with detailed
formal specifications and additional assertions
such as invariants.
VerCors~\cite{VerCors} and VeriFast~\cite{VeriFast}
are two prominent examples of deductive verifiers for Java,
both supporting annotations written in a combination of JML~\cite{JML}
and separation logic.
VerCors offers a higher degree of automation, as it
relies on the Viper intermediate verifier~\cite{Viper} to discharge verification conditions,
while VeriFast supports interactive correctness proofs.

\nicepar{Model checking\;}
is a popular verification technique
based on analyzing (finite-)state models
against temporal-logic properties.
Tools like JPF (Java Pathfinder)~\cite{JPF}, JayHorn~\cite{JayHorn},
JBMC~\cite{JBMC}, and JMC~\cite{JMC}
are \emph{software} model checkers:
they encode the semantics of a Java program using a 
state model that can be analyzed by a model checker such as Spin~\cite{Spin}.
This provides a high degree of automation, while still supporting
user-defined properties (usually expressible in temporal logic).
\tech follows a similar approach,
but it leverages Petri nets (instead of less expressive automata models)
to naturally model several aspects of a program's concurrent behavior.

A key issue when designing an analysis based on model checking
is that the full semantics of a (Java) program is generally infinite-state.
Some tools (e.g., JPF, JBMC, JMC)
perform a \emph{bounded} exploration
of the infinite program state; hence, they are powerful testing tools, but
are not sound in general.
JPF, in particular, is a mature analysis framework based on a
custom controllable version of the Java virtual machine;
it can systematically or randomly explore execution paths
and different thread interleavings.
Symbolic JPF (SPF) performs the state exploration \emph{symbolically}
(using a form of symbolic execution),
so that it is systematic and satisfies complex coverage criteria.

Other techniques build
a finitely-analyzable \emph{over-approximation} of a program's state space,
which loses precision but retains soundness.
Techniques like Infer's RacerD~\cite{Infer,RacerD}
(based on separation logic and bi-abduction),
Checkmate~\cite{Checkmate} (based on abstract interpretation)
and Chord~\cite{Chord}
(an unsound, precise technique based on context-sensitive analyses)
are all different applications of static analysis to detect data races.
JaDA~\cite{JaDA} offers deadlock detection for (a subset of) Java bytecode;
it uses typing rules to define an infinite-state abstract model of
the program's lock dependencies;
the model is analyzable by means of a fixpoint decision algorithm.
\tech also builds, by means of a dataflow analysis,
an approximation of a Java program's executions,
which it encodes as a finitely-analyzable Petri net.

\nicepar{Petri nets\;}
are a classic model of concurrency.
To our knowledge, they have been historically mainly used to build abstract,
high-level models, to encode the semantics
of core concurrency properties
(e.g., causal atomicity~\cite{CausalAtomicity})
and primitives
(e.g., synchronization through conditional variables~\cite{KTH}).
In our work,
we leverage the recent advances in model-checking tools for Petri nets~\cite{Tina,LoLA,MCC}
and use them as back-end of \tool.

\nicepar{Dynamic analysis\,} 
has become more popular in recent years to detect concurrency bugs
such as deadlocks~\cite{DeadlockPrediction},
data races~\cite{MathurRace,Kulkarni21}, and linearizability~\cite{lincheck,MathurLinearizability} and atomicity~\cite{MathurAtomicity,RegionTrackAtomicity} violations.
While such techniques do not offer soundness,
they are practical (i.e., since they are based on executing a program, they support all language features) and scalable (e.g., they analyze a trace in linear time).

\nicepar{Soundness vs.\ precision.}
On paper, static and dynamic analysis offer complementary
advantages and disadvantages: static techniques overapproximate program behavior,
and hence their results are sound (exhaustive) but imprecise;
dynamic techniques underapproximate program behavior,
and hence their results are precise (no false alarms) but unsound.
In practice, the boundary between soundness and precision is somewhat fuzzy,
and even static techniques are very often unsound in certain cases~\cite{soundiness},
because they may not fully support certain language features due to practical concerns.
As we will discuss in \autoref{sec:approach}, \tech follows a similar approach of
being ``mostly'' sound and ``reasonably'' precise:
\begin{enumerate*}
\item \tech's implementation does not (fully) support certain Java language features; thus,
the analysis results may be unsound for programs that use those features.
\item \tech relies on an alias analysis to identify possibly shared lock and thread variables;
  when the alias analysis results are insufficiently precise,
  \tech's own analysis may in turn become unsound or imprecise.
\end{enumerate*}
Ultimately, the practicality of \tech in analyzing concurrent programs
is established with \autoref{sec:experiments}'s experimental evaluation,
which highlights its capabilities and limitations, also in comparison
to other similar concurrency analysis tools.

\section{Preliminaries: Jimple Bytecode, \rb, and Petri nets}

\tech encodes the concurrent dataflow semantics of a JVM bytecode program
as a Petri net.
This section introduces
the intermediate representations used by \tech.

\begin{figure}[!tb]
  \centering\small
\begin{adjustwidth}{-12mm}{0mm}
  \begin{subfigure}{0.68\textwidth}
    \centering
  \begin{alignat*}{2}
    \nt{Jimple} & ::=\ \nt{Method}^* \\
    \nt{Method} & ::=\ \Ji{method}\ \nt{Id} \colon \nt{Type}^* \to \nt{Type}\ \{ \nt{Instr}^* \}\\
    \nt{Instr} & ::=\
                       \Ji{noop}
                       \ \mid\ \Ji{goto}\:\nt{Label}
                       \ \mid\ \Ji{if}\:\nt{Var}\;\nt{Label} \ \mid \\
                & \qquad
                  \Ji{switch}\:\nt{Var}\;(\nt{Var}\colon\nt{Label})^*
                  \ \mid \\
                & \qquad
                  \nt{Var}\:\Ji{:=}\:\nt{Expr}
                  \ \mid\ \Ji{return} \;\nt{Var}
                  \ \mid\ \Ji{return} \ \mid \\
                & \qquad
                  \Ji{monitor_enter}\;\nt{Var} \ \mid\ \Ji{monitor_exit}\;\nt{Var} \\
    \nt{Expr} & ::=\ 
                \nt{Var}\, \oplus \nt{Var}
                \ \mid\ \odot \nt{Var}
                \ \mid\ \nt{Const}
                \ \mid\ \nt{Call} \\
    \nt{Call} & ::=\
                \Ji{invoke}\:\nt{Id}\:\nt{Var}^*
  \end{alignat*}
  \caption{Syntax of Jimple.}
  \label{fig:syntax-jimple}
\end{subfigure}
\begin{subfigure}{0.3\textwidth}
    \centering
  \begin{alignat*}{2}
    \nt{RB} & ::=\ \nt{Proc}^* \\
    \nt{Proc} & ::=\
                     \RB{def}\;\nt{Id}\;\RB{begin}\:\nt{Cmd}^*\:\RB{end}\\
    \nt{Cmd} & ::=\
               \RB{skip} \ \mid\ \\
    & \qquad
      \RB{goto}\:\nt{Label}
      \ \mid\ \RB{jump}\:\nt{Label}\:\nt{Label}\ \mid\ \\
            & \qquad
              \RB{read}\:\nt{Var}
              \ \mid\ \RB{write} \:\nt{Var}\ \mid\ \\
            & \qquad
              \RB{acquire}\;\nt{Var} \ \mid\ \RB{release}\;\nt{Var}\ \mid\ \\
            & \qquad
              \RB{fork}\;\nt{Id}\:\nt{Thread}\ \mid\ \RB{join}\;\nt{Thread}\ \mid\ \\
            & \qquad
              \RB{call}\;\nt{Id}
  \end{alignat*}
  \caption{Syntax of \rb.}
  \label{fig:syntax-rb}
\end{subfigure}
\end{adjustwidth}
\caption{Syntax of the intermediate representations used by \petrify.}
  \label{fig:syntax}
\end{figure}

\subsection{The Jimple Bytecode Representation}
\label{sec:jimple}

  Rather than working directly with JVM bytecode,
\petrify uses Soot's Jimple bytecode representation:
a typed bytecode form
that abstracts several low-level details and incorporates static information.
\autoref{fig:syntax-jimple} outlines the main constructs of Jimple,
using a simplified abstract syntax.
  Before going through it,
  we point out that
Jimple uses an SSA (static single assignment) form;
thus, any complex expression $((v_0 \oplus_1 v_1) \oplus_2 v_2) \cdots \oplus_n v_n)$
becomes a sequence of assignments to local variables
$r_1 := v_0 \oplus_1 v_1$, $r_2 := r_1 \oplus_2 v_2$, \ldots, $r_n := r_{n-1} \oplus_n v_n$, such that $r_n$ stores the value of the whole expression.
Thus, Jimple instructions generally only take \emph{variables} (not expressions)
as arguments.

A \nt{Jimple} program is a collection of methods,
with at least one \Ji{main} method
that corresponds to the program's entry point.
A \nt{Method} has a name, a typed signature, and a body
consisting of a sequence of instructions,
which include:
\begin{enumerate*}
\item \emph{Control flow} instructions:
  unconditional (\Ji{goto}) and conditional (\Ji{if}, \Ji{switch})
  jumps to a location with a given label within the method;
  \Ji{return} to the caller.
\item \emph{Assignment} instructions
  perform any kind of expression evaluation using the SSA form discussed above.
\item \emph{Monitor} instructions
  correspond to the synchronization
  when entering (\Ji{monitor_enter}) and exiting a \J{synchronized} block.
\item \emph{Call} instructions
  are also only used in the right-hand side
  of an assignment.
  For simplicity, \Ji{invoke} represents all five variants of
  call instructions (\Bc{static}, \Bc{virtual}, etc.)
  available in bytecode.
\item The \Ji{noop} instruction does nothing.
\end{enumerate*}

\begin{expara}
  \autoref{fig:ex:java}'s example in Jimple consists of three methods:
  \Ji{main},
  and two anonymous methods for the \J{Runnable} objects.
  Each \J{synchronized} block is a pair of matching
  \Ji{monitor_enter} and \Ji{monitor_exit} instructions.
  The other statements are different variants of \Ji{invoke}:
  \Ji{invokespecial} for the \J{new} creation expressions,
  and \Ji{invokevirtual} for the \J{start()} calls.
\end{expara}

\subsection{The \rb Intermediate Language}
\label{sec:rb}

To streamline the encoding of concurrent behavior
into PNs,
we introduce the \rb intermediate language.
In a nutshell, \rb (short for Rock Bottom)
is a simplified bytecode-like representation
that captures flow- and context-sensitive information
while abstracting away path-sensitive details.

\paragraph{Syntax of \rb.}
\autoref{fig:syntax-rb} outlines the syntax of \rb.
An \nt{RB} program is a collection of procedures;
like in Jimple (and Java)
we assume that at least one \Ji{main} procedure exists.
A \nt{Proc} has a unique name, and a body
consisting of a sequence of commands
marked by \RB{begin} and \RB{end}.
We also assume that each procedure $p$ also defines
labels $\RB{entry}_p$ and $\RB{exit}_p$ marking, respectively, $p$'s
unique entry and exit points.
\rb commands include:
\begin{enumerate*}
\item \emph{Control flow} commands:
  unconditional (\RB{goto}) and nondeterministic (\RB{jump})
  jumps;
  synchronously \RB{call} a procedure $p$;
  spawn (\RB{fork}) and wait for (\RB{join})
  a parallel thread $t$ running a procedure $p$.
  We assume that the thread identifiers in \RB{fork} and \RB{join} commands
  are distinct from all other identifiers,
  and that the \RB{main} procedure runs on thread $t_0$.
\item \emph{Synchronization} commands
  to \RB{acquire} and \RB{release} a lock.
\item \emph{Access} commands
  to \RB{read} and \RB{write} a variable.
\item The \RB{skip} command does nothing.
\end{enumerate*}
Given an \rb program \nt{RB},
$T$ denotes the set of all \nt{Thread} identifiers,
$P$ the set of all \nt{Procedure} identifiers,
$B$ the set of all \nt{Label}s
(we assume that each command has a unique label),
and
$V$ the set of all \nt{Variable} names
mentioned anywhere in the program.

\begin{figure}[t]
  \centering
  \begin{adjustwidth}{-10mm}{-3mm}
  \begin{align*}
    \inferrule{
    s = (t, p, \ell\colon c, L, R, \running) \in S \\\\
    c \in \{ \RB{skip}, \RB{read}\: v, \RB{write}\: v\}
    }{
    S' = S \setminus \{s\} \cup \{(t, p, \ell + 1, L, R, \running) \}
    }
    &&
    \inferrule{
    s = (t, p, \ell\colon \RB{call}\:p', L, R, \running) \in S
    }{
    S' = S \setminus \{s\} \cup \{(t, p', \RB{entry}_{p'}, L, R + [\ell + 1], \running) \}
    }
    \\
    \inferrule{
    s = (t, p, \RB{exit}_p, L, R + [\ell], \running) \in S
    \\ \ell \in p'
    }{
    S' = S \setminus \{s\} \cup \{(t, p', \ell, L, R, \running) \}
    }
    &&
    \inferrule{
    s = (t, p, \RB{exit}_p, L, \emptyset, \running) \in S
    }{
    S' = S \setminus \{s\} \cup \{(t, p, \RB{exit}_p, L, \emptyset, \terminated) \}
    }
    \\
    \inferrule{
    s = (t, p, \ell\colon \RB{goto}\:\ell', L, R, \running) \in S
    }{
    S' = S \setminus \{s\} \cup \{(t, p, \ell', L, R, \running) \}
    }
    &&
    \inferrule{
       s = (t, p, \ell\colon \RB{jump}\:\ell_1 \ell_2, L, R, \running) \in S \quad
      \ell' \in \{ \ell_1, \ell_2 \}
    }{
    S' = S \setminus \{s\} \cup \{(t, p, \ell', L, R, \running) \}
    }
    \\
    \inferrule{
    s = (t, p, \ell\colon \RB{acquire}\:k, L, R, \running) \in S
    \\ k \not\in L
    \\\\
    \forall t', L' \cdot (t', \_, \_, L', \_, \_) \in S \land t' \neq t \rightarrow k \not\in L'
    }{
    S' = S \setminus \{s\} \cup \{(t, p, \ell + 1, L \cup \{k\}, R, \running) \}
    }
    &&
    \inferrule{
       s = (t, p, \ell\colon \RB{release}\:k, L, R, \running) \in S
       \\ k \in L
    }{
    S' = S \setminus \{s\} \cup \{(t, p, \ell + 1, L \setminus \{k\}, R, \running) \}
    }
    \\
    \inferrule{
       s = (t, p, \ell\colon \RB{fork}\:p'\:t', L, R, \running) \in S
    }{
    S' = S \setminus \{s\} \cup \{(t, p, \ell + 1, L, R, \running) \}\\\\
    \qquad\qquad\qquad\quad\cup\ \{(t', p', \RB{entry}_{p'}, \emptyset, \emptyset, \running) \}
    }
    &&
    \inferrule{
       s = (t, p, \ell\colon \RB{join}\:t', L, R, \running) \in S
\\\\
    s' = (t', p', \ell', L', R', \terminated) \in S
    }{
    S' = S \setminus \{s, s'\} \cup \{(t, p, \ell + 1, L, R, \running) \}
    }
  \end{align*}
\end{adjustwidth}
\caption{Operational semantics of \rb.
    Each rule describes one step of evaluation $S \leadsto S'$ when
    executing the command in the rule's premise.}
  \label{fig:rb-semantics}
\end{figure}

\paragraph{Semantics of \rb.}
The \emph{state} $S$ of an \rb program
is a set of tuples $(t, p, \ell, K, R, \tau)$,
where $t \in T$ is a thread,
$p \in P$ is a procedure,
$\ell \in B$ is a command label,
$K \subseteq V$ is a set of locked variables,
$R \in B^*$ is a sequence of procedure/label pairs,
and $\tau \in \{\running, \terminated\}$ is the thread's termination state.
Informally,
such a tuple denotes a thread $t$ is ready to run command at label $\ell$
in procedure $p$, while holding locks $K$;
$\tau$ denotes whether $t$
is still running (\running) or has terminated (\terminated);
and $R$ is a stack of return locations of pending calls.
The \emph{initial state} of an \rb program
is $\{(t_0, \RB{main}, \RB{entry}_{\RB{main}}, \emptyset, \emptyset, \running)\}$.

\autoref{fig:rb-semantics} outlines an operational semantics of \rb.
Each rule shows how the state $S$ changes after executing a different command.
Command \RB{skip}
simply continues execution of the thread $t$ to the next location $\ell + 1$.
Since \rb abstracts away path-sensitive details,
\RB{read} and \RB{write} behave exactly like \RB{skip}.
However, we still include them as separate commands since they allow us
to model, in \rb, the different interleavings of read and write operations
in a JVM program, and behaviors that depend on them---such as data races.
A $\RB{call}\: p'$ command appends the return label $\ell + 1$ to $R$ and continues executing the callee $p'$ from its entry point.
Conversely, when a call terminates, the most recent return label is removed from $R$ and used as next command to execute.
If the execution of a procedure $p$ reaches its exit point and $R$ is empty,
this makes the whole thread $t$ terminate---denoted by $\tau = \terminated$.
Both kinds of branching commands, \RB{goto} and \RB{jump},
change the current command label: \RB{goto} does so unconditionally,
whereas \RB{jump} nondeterministically picks one of two possible labels.
A thread $t$ can execute an $\RB{acquire}\:k$ only if
any other thread $t'$ is \emph{not} holding a lock on $k$.
In contrast, $t$ can execute a $\RB{release}\:k$ without waiting,
as long as it is currently holding a lock on $k$.
Finally, a $\RB{fork}\:p'\:t'$
adds a new tuple $(t', p', \RB{entry}_{p'}, \emptyset, \emptyset, \running)$
to the state $S$, corresponding to a new thread $t'$ starting to execute procedure $p'$.
And a $\RB{join}\:t'$ can execute only when thread $t'$
has terminated.

\begin{expara}
When \autoref{fig:ex:rb}'s example \rb program
deadlocks, its state $S$ consists of the tuples:
{\small $(\RB{t_12}, \RB{r_12}, \ell_7, \{\RB{lock1}\}, [], \running)$},
{\small $(\RB{t_21}, \RB{r_21}, \ell_{11}, \{\RB{lock2}\}, [], \running)$},
{\small $(t_0, \RB{main}, \RB{exit}_{\RB{main}}, \emptyset, [], \terminated)$}.
\linebreak
Neither $\ell_7$ nor $\ell_{11}$ can execute,
because each thread state invalidates the other's progress precondition.
\end{expara}

\subsection{Petri Nets}
\label{sec:petri}

A Petri net is a tuple $(\Pi, \Delta, A, I)$,
where $\Pi$ is a set of \emph{places},
$\Delta$ is a set of \emph{transitions},
$A \subseteq (\Pi \times \Delta) \cup (\Delta \times \Pi)$ is a set of \emph{arcs}
connecting transitions to places and places to transitions,
and $I \colon \Pi \to \naturals$ is an \emph{initial marking}.
The \emph{preset} $\nt{pre}(\delta)$ of a transition $\delta \in \Delta$
is the set of places $\pi$ such that $(\pi, \delta) \in A$;
and the \emph{postset} $\nt{post}(\delta)$
is the set of places $\pi$ such that $(\delta, \pi) \in A$.
It is customary to picture a PN as a graph where
places are circles, transitions are boxes, and arcs are arrows.
\autoref{fig:motivating-example} displays the running example's
PN using this notation; the black disks denote
the net's initial marking.

\paragraph{Semantics of Petri nets.}
The \emph{state} of a PN is a \emph{marking} $m \colon \Pi \to \naturals$,
which denotes how many \emph{tokens} $m(\pi) \geq 0$ each place $\pi \in \Pi$ holds.
A transition $\delta \in \Delta$ is \emph{enabled} in a state $m$
whenever $m(\pi) > 0$ for every place $\pi$ in $\nt{pre}(\delta)$;
in other words, all places that connect to $t$ are marked with at least one token.
Whenever a transition $\delta$ is enabled, it can nondeterministically \emph{fire}.
If $\delta$ fires when the net is in state $m$,
the new marking $m'$ is such that $m'(p) = m(p) - 1$ for every place $p \in \nt{pre}(\delta)$, and
$m'(p) = m(p) + 1$ for every place $p \in \nt{post}(\delta)$.
In other words, each place in $\delta$'s preset consumes one token,
whereas each place in $\delta$'s postset acquires one additional token.
If no transition is enabled in a marking $m$, the PN is in a \emph{deadlock}.

According to these definitions,
a PN determines a set of possible \emph{firing sequences},
corresponding to all sequences of transitions that can be triggered from
the initial marking;
each firing sequence corresponds to a sequence of markings
$m_0 \vdash m_1 \vdash \ldots$ that begins
with $m_0 = I$.
The set of all firing/marking sequences denotes a PN's \emph{semantics}.

\begin{expara}
  From its initial marking,
  \autoref{fig:ex:pn}'s PN may reach the state
  where places
  $(\RB{t_12}, \ell_7)$, $(\RB{t_21}, \ell_{11})$, $(\RB{t_12}, \RB{lock1})$, $(\RB{t_21}, \RB{lock2})$
  are marked.
  From that state, no transition is possible:
  place $(\RB{t_12}, \RB{lock2})$ is unmarked,
  which disables transition $(\RB{t_12}, \ell_{7})$;
  and place $(\RB{t_21}, \RB{lock1})$ is unmarked,
  which disables transition $(\RB{t_21}, \ell_{11})$.
\end{expara}

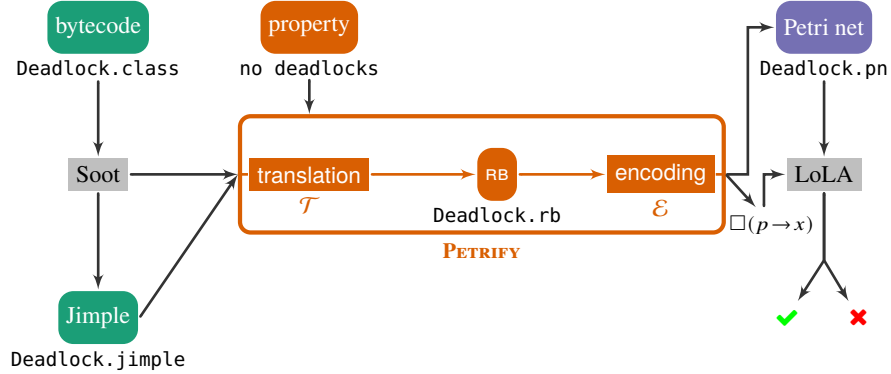
\begin{figure}[!bt]
\begin{adjustwidth}{-30mm}{-30mm}
\centering
\begin{tikzpicture}[
  databox/.style 2 args={rectangle,very thick,
    rounded corners=2mm,font=\footnotesize#1,
    minimum width=5mm,minimum height=7mm,
    label={[below=20pt]#2}},
  toolbox/.style={rectangle,very thick,font=\footnotesize#1},
  align=center
  ]

  \matrix[row sep=8mm,column sep=5mm] {
\node[databox={}{\texttt{Deadlock.class}},fill=jimplecol,text=white] (bytecode) {bytecode};
    & \node[databox={}{\texttt{no deadlocks}},fill=rbcol,text=white] (prop) {property};
    &&&
    \node[databox={}{\texttt{Deadlock.pn}},fill=pncol,text=white] (petrinet) {Petri net};
    \\
    \node[toolbox,fill=graycol,text=black] (soot) {Soot};
    &
    \node[toolbox={\sffamily},label={[below=12pt]\textcolor{rbcol}{$\translation$}},fill=rbcol,text=white] (analysis) {translation};
    &
    \node[databox={\sffamily}{\texttt{Deadlock.rb}},fill=rbcol,text=white] (rb) {\rb};
    &
    \node[toolbox={\sffamily},label={[below=14pt]\textcolor{rbcol}{$\encoding$}},fill=rbcol,text=white] (encoding) {encoding};
    &
    \node[toolbox,fill=graycol,text=black] (pnanalyzer) {LoLA};
    \\
    \node[databox={}{\texttt{Deadlock.jimple}},fill=jimplecol,text=white] (jimple) {Jimple};
    &&&&
    \node (output-mid) {};
    \\
  };

  \node[left=3pt of output-mid] (ok) {\color{green}{\faCheck}};
  \node[right=3pt of output-mid] (fail) {\color{red}{\faClose}};
  \node[below left=2mm and -5mm of pnanalyzer] (tl) {\scriptsize$\Box(p\!\to\!x)$};

  \node [fit=(analysis)(encoding),draw=petrifycol,
         ultra thick,rounded corners,label={[petrifycol]south:\textbf{\tech}},
         inner ysep=5mm] (petrify) {};

  \begin{scope}[color=petrifycol,line width=1pt]
    \draw (petrify.west) -- (analysis);
    \draw (encoding.east) -- (petrify.east);
    \begin{scope}[round cap-latex']
      \draw (analysis) -- (rb);
      \draw (rb) -- (encoding);
    \end{scope}
  \end{scope}

  \begin{scope}[color=black!80,line width=1pt,round cap-latex',every
    node/.style={font=\footnotesize}]
    \draw ($(bytecode)+(0,-7mm)$) -- (soot);
    \draw (soot) -- (jimple);
    \draw (soot) -- (petrify);
    \draw (jimple.east) -- (petrify.west);
    \draw (petrify.east) -- ($(petrify.east)!0.37!(pnanalyzer.west)$) |- (petrinet.west);
    \draw (petrify.east) -- (tl);
    \draw ($(petrinet)+(0,-7mm)$) -- (pnanalyzer);
    \draw (pnanalyzer) -- ($(pnanalyzer)!0.6!(output-mid)$) -- (ok);
    \draw (pnanalyzer) -- ($(pnanalyzer)!0.6!(output-mid)$) -- (fail);
    \draw ($(tl.north) + (-3pt,0)$) |- (pnanalyzer);
    \draw ($(prop.south) + (0,-3.1mm)$) -- (prop.south |- petrify.north);
  \end{scope}
  
\end{tikzpicture}
\end{adjustwidth}
\caption{An overview of how \tech works.}
\label{fig:workflow}
\end{figure}

\section{How \tech Works}
\label{sec:approach}

This section details \tech's approach.
As outlined in \autoref{fig:workflow},
\tech inputs a program $J$ in bytecode;
  precisely, it targets the Jimple format offered by the Soot static analyzer.
\tech translates $J$ into
a corresponding \rb program $R = \translation(J)$
(\autoref{sec:jimple-to-rb});
it then encodes $R$'s semantics as a Petri net $P = \encoding(R)$
(\autoref{sec:rb-to-pn});
it also produces suitable temporal logic
properties that can be passed as input, together with $P$, to a PN analyzer
(e.g., LoLA) to detect concurrency issues
(\autoref{sec:encoding-props}).

\subsection{Translation of Jimple into \rb}
\label{sec:jimple-to-rb}
\label{sec:rb-soundness-precision}

Given a bytecode program (in Jimple format) $J$,
\tech builds an \rb program $\translation(J) = R$
that captures a flow- and context-sensitive, path-insensitive approximation
of $J$'s behavior.

\subsubsection{Modeling Aliasing.}\label{sec:aliasing}

Consider any identifier $n$ that appears in $J$, including
variable identifiers and method names.
Let $\eta(n)$ denote $n$ expressed in a form that is unambiguous throughout the program:
for example, if $\Ji{o.m}$ denotes a method \Ji{m} of an object \Ji{o} of class \Ji{O},
$\eta(\Ji{o.m})$ denotes the fully qualified name \Ji{O.m}
with a suffix that distinguishes it from other overloaded variants.
Let $\varphi \triangleq \bigcup_{\iota\colon \nt{Var}} \{ \eta(\iota) \}$
denote the set of unique variable identifiers anywhere in $J$---which we'll call
``vars'' for short.
And let $\locs$ denote the set of all locations in $J$.
\tech's translation of bytecode uses 
a \emph{may point to} analysis as follows.
Let $\iota \in \varphi$ be a var
used at location $\ell \in \locs$ in $J$;
$\mu_{\iota, \ell} \subseteq \varphi$ denotes the \emph{set} of
vars that may be aliased to $\iota$
according to the may point to analysis.
The may-alias sets $\mu_{\iota, \ell}$ determine an equivalence relation $\simeq^+\ \subseteq(\varphi \times \locs) \times (\varphi \times \locs)$ among
vars as follows:
first, $\simeq\ \subseteq(\varphi \times \locs) \times (\varphi \times \locs)$ is the reflexive and symmetric relation defined by 
$\iota_1, \ell_1 \simeq \iota_2, \ell_2$ iff $\mu_{\iota_1, \ell_1} \cap \mu_{\iota_2, \ell_2} \neq \emptyset$;
then, $\simeq^+$ is the transitive closure of $\simeq$: $\iota, \ell \simeq^+ \iota', \ell'$ iff
there exist $\iota_1, \ell_1,\ldots,\iota_n, \ell_n$ such that $\iota_1 = \iota$, $\ell_1 = \ell$, $\iota_n = \iota'$, $\ell_n = \ell'$, and
$\iota_k, \ell_k \simeq \iota_{k+1}, \ell_{k+1}$ for all $1 \leq k < n$.
Finally, let $\alpha(\iota, \ell)$ be a unique identifier
that corresponds to the equivalence class of $\varphi$ according to $\simeq^+$
that $\iota, \ell$ belongs to.
Intuitively, $\alpha$ assigns the same identifiers to any two vars in $J$
iff they may be aliased.
\iflong
This approach is parametric with respect to the may point to analysis that is used.
\fi

\subsubsection{\rb Encoding.}\label{sec:rb-encoding}

\autoref{tab:jimple-to-rb}
outlines how \tech translates Jimple statements into \rb commands.
Methods become procedures in \rb, whose bodies get translated statement by statement.
\iflong
The translation of \Ji{noop}, \Ji{goto}, and \Ji{return}
is straightforward.
\fi
Conditional jumps (\Ji{if} and \Ji{switch})
become nondeterministic jumps (\RB{jump}) in \rb.
This makes all control-flow paths in the Jimple program feasible in the \rb program.
Lock acquisition statements (\Ji{monitor_enter} and \Ji{monitor_exit})
become the corresponding commands in \rb (\RB{acquire} and \RB{release}).
Finally, an assignment translates to a translation of its right-hand side expression,
followed by a \RB{write} to its target.
The translation of non-call expressions is straightforward,
as it corresponds to a read of the variables involved in the expression.
An \Ji{invoke} translates to \RB{read}s of its actual arguments
and possibly its target, followed by a \RB{call} to the \rb procedure $\eta(o.m)$ translating the called method $o.m$.
An exception is for calls of methods \Ji{start} and \Ji{join}
on targets of type \Ji{Thread}: these become \RB{fork} and \RB{join}
respectively.

\begin{table}[!tb]
  \centering
  \begin{tabular}{ll}
    \toprule
    \multicolumn{1}{c}{\textsc{Jimple} $s$} & \multicolumn{1}{c}{\rb \textsc{translation} $\translation(s)$}
    \\
    \midrule
    $\Ji{method}\ m \colon t_1 \ldots t_n \to t\ \{ B \}$
    & $\RB{def}\ \eta(m(t_1\ldots t_n))\ \RB{begin}\ \translation(B)\ \RB{end}$
    \\
    $s_1 ; s_2$ & $\translation(s_1) ; \translation(s_2)$
    \\
    $\Ji{noop}$ & $\RB{skip}$
    \\
    $\Ji{goto}\ \ell'$ & $\RB{goto}\ \bar{\ell'}$
    \\
    $\Ji{if}\ v\ \ell'$ & $\RB{read}\ \alpha(v) ; \RB{jump}\ \nextloc\ \bar{\ell'}$
    \\
    $\Ji{switch}\ v\ (v_1 \colon \ell_1) \ldots (v_n \colon \ell_n)$ & $\RB{read}\ \alpha(v) ; \RB{read}\ \alpha(v_1) ; \RB{jump}\nextloc \bar{\ell_1};\ldots;\RB{read}\ \alpha(v_n) ; \RB{jump}\nextloc \bar{\ell_n}$
    \\
    $v\ \Ji{:=}\ e$ & $\translation(e) ; \RB{write}\ v$
    \\
    $\Ji{return}\ v$ & $\RB{read}\ \alpha(v) ; \RB{goto exit}$
    \\
    $\Ji{return}$ & $\RB{goto exit}$
    \\
    $\Ji{monitor_enter}\ v$ & $\RB{acquire}\ \alpha(v)$
    \\
    $\Ji{monitor_exit}\ v$ & $\RB{release}\ \alpha(v)$
    \\
    \cmidrule(lr){1-2}
    $c \colon \nt{Const}$ & --
    \\
    $v_1 \oplus v_2$ & $\RB{read}\ \alpha(v_1) ; \RB{read}\ \alpha(v_1)$
    \\
    $\odot\ v$ & $\RB{read}\ \alpha(v)$
    \\
    $\Ji{invoke}\ m\ o\ a_1\ldots a_n$
    & $\RB{read}\ \alpha(v_1) ;\ldots; \RB{read}\ \alpha(v_n) ; \RB{read}\ \alpha(o) ; \RB{call}\ \eta(o.m)$
    \\                                                          
    $\Ji{invoke}\ \Ji{start}\ t\colon \Ji{Thread}$
                       & $\RB{fork}\ \eta(t.\Ji{start})\ \alpha(t)$
    \\
    $\Ji{invoke}\ \Ji{join}\ t\colon \Ji{Thread}$
                       & $\RB{join}\ \alpha(t)$
    \\
    \bottomrule
  \end{tabular}
  \caption{Translation $\translation$ of Jimple instructions (top) and expressions (bottom) into \rb commands.
    $\bar{\ell}$ denotes the location in the \rb program that translates
    the instruction at location $\ell$ in the Jimple program;
    $\nextloc$ denotes the location of the next command;
    $\eta(n)$ denotes a unique form of identifier $n$;
    and $\alpha(v)$ denotes a group of variables that may be aliased
    (see the text for a precise definition).
  }
  \label{tab:jimple-to-rb}
\end{table}

\subsubsection{Soundness and Precision of the \rb Encoding.}

The translation $R = \translation(J)$ of $J$ is \emph{sound} if the feasible execution paths in $R$ are a \emph{superset}
of those in $J$; conversely, an unsound translation may omit error paths that are possible in the original program.
The translation is \emph{precise} if the feasible execution paths in $R$ are a \emph{subset}
of those in $J$; conversely, an imprecise translation may include spurious error paths that are impossible in the original program.
\iflong
Based on these standard definitions, l\else L\fi et's summarize how the translation scheme affects soundness and precision.
\begin{enumerate*}
\item Translating conditionals with
  nondeterministic \RB{jump}s is sound but path-insensitive; hence, it
  generally involves a loss of \emph{precision}.

\item The translation of lock operations in $J$ with \RB{acquire} and \RB{release}
  in $R$ is sound provided the lock is not used \emph{reentrantly}.
  According to \autoref{fig:rb-semantics}, \rb does not allow a thread
  to acquire a lock on $k$ if it already holds a lock on it;
  thus, $R$ may omit such executions even if they are possible in $J$.\footnote{
    Soundly modeling reentrant locks would require a stack-like counting
    mechanism, which goes beyond the expressiveness of plain PNs;
    hence, it belongs to future work.
  }

\item \tech's sound translation of call instructions depends on
  whether Soot can retrieve the body of the invoked closure object
  in \Ji{invokedynamic} instructions.
  This is possible in simple cases such as
  \autoref{fig:motivating-example}'s example, where \Ji{invokedynamic} is used to execute
  the \J{Runnable} lambdas \J{r_12} and \J{r_21}; more complex
  instances of \Ji{invokedynamic} would become \RB{skip} in $R$,
  which introduces unsoundness in general.

\item Other bytecode instructions that
  are not listed in \autoref{fig:syntax-jimple} are currently
  unsupported by \tech.
  The translation replaces any unsupported instruction $I$
  with a \RB{skip}, which means that $R$ doesn't model $I$'s semantics.
  This may result in a loss of soundness or precision, depending on what execution paths the unsupported instruction
  does enable or block.

\item Two variables $v_1, v_2$ that may be \emph{aliased} in $J$ are lumped together
  into a single variable $v = \alpha(v_1) = \alpha(v_2)$ in $R$;
  this may introduce a loss of \emph{soundness} or, more commonly, \emph{precision}.
 
\end{enumerate*}
In practice, these limitations mainly imply the lack of support for certain
program features.
As our experiments in \autoref{sec:experiments} show,
\tech remains applicable on broad range of Java programs following
different concurrency patterns and features.
While we plan to remove some limitations in future work,
\tech's capabilities are consistent with the pragmatic
approach of making program analysis work with realistic programs
despite theoretical limitations~\cite{soundiness}.

\begin{expara}
  \autoref{fig:ex:rb} shows \tech's \rb translation of \autoref{fig:ex:java};
  for simplicity, \autoref{fig:ex:rb} omits the \RB{write}s of variables
  \RB{r_12}, \RB{r_21}
  (assignments on lines 4, 5) and the \RB{read}s of variables \RB{r_12}, \RB{r_21}
  (\J{new} expressions on lines 6, 7).
  Since the program does not have data-dependent path conditions,
  and there is no aliasing among variables,
  the translation to \rb is sound \emph{and} precise.
  The \rb program consists of three procedures: \RB{main},
  and two \J{Runnable} anonymous functions that
  are each thread's \J{run()} method.
\end{expara}

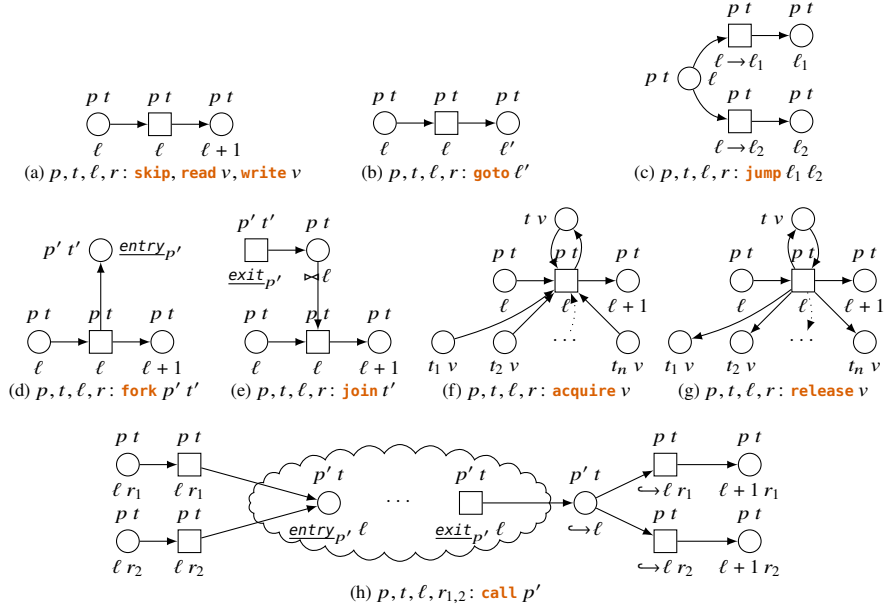
\begin{figure}[!bth]
  \centering
  \begin{subfigure}{0.3\textwidth}
    \centering
    \scriptsize
    \begin{tikzpicture}[-latex]
      \pnplace{p-p-t-ell}{$p$}[$t$][$\ell$]
      \pntrans{t-p-t-ell}{$p$}[$t$][$\ell$][right=of p-p-t-ell]
      \pnplace{p-p-t-ell-1}{$p$}[$t$][$\ell + 1$][right=of t-p-t-ell]
      \draw (p-p-t-ell) to (t-p-t-ell);
      \draw (t-p-t-ell) to (p-p-t-ell-1);
    \end{tikzpicture}
    \caption{$p, t, \ell, r \colon \RB{skip}, \RB{read}\ v, \RB{write}\ v$}
    \label{fig:rb2pn:skip}
  \end{subfigure}
  \begin{subfigure}{0.3\textwidth}
    \centering
    \scriptsize
    \begin{tikzpicture}[-latex]
      \pnplace{p-p-t-ell}{$p$}[$t$][$\ell$]
      \pntrans{t-p-t-ell}{$p$}[$t$][$\ell$][right=of p-p-t-ell]
      \pnplace{p-p-t-ell-1}{$p$}[$t$][$\ell'$][right=of t-p-t-ell]
      \draw (p-p-t-ell) to (t-p-t-ell);
      \draw (t-p-t-ell) to (p-p-t-ell-1);
    \end{tikzpicture}
    \caption{$p, t, \ell, r \colon \RB{goto}\:\ell'$}
    \label{fig:rb2pn:goto}
  \end{subfigure}
  \begin{subfigure}{0.3\textwidth}
    \centering
    \scriptsize
    \begin{tikzpicture}[-latex]
      \pnplace*{p-p-t-ell}{$p$}[$t$][$\ell$]
      \pntrans{t-p-t-ell-1}{$p$}[$t$][$\ell\!\to\!\ell_1$][above right=3mm and 4mm of p-p-t-ell]
      \pntrans{t-p-t-ell-2}{$p$}[$t$][$\ell\!\to\!\ell_2$][below right=3mm and 4mm of p-p-t-ell]
      \pnplace{p-p-t-ell-1}{$p$}[$t$][$\ell_1$][right=of t-p-t-ell-1]
      \pnplace{p-p-t-ell-2}{$p$}[$t$][$\ell_2$][right=of t-p-t-ell-2]
      \draw (p-p-t-ell) to[bend left] (t-p-t-ell-1);
      \draw (p-p-t-ell) to[bend right] (t-p-t-ell-2);
      \draw (t-p-t-ell-1) to (p-p-t-ell-1);
      \draw (t-p-t-ell-2) to (p-p-t-ell-2);
    \end{tikzpicture}
    \caption{$p, t, \ell, r \colon \RB{jump}\:\ell_1\:\ell_2$}
    \label{fig:rb2pn:jimp}
  \end{subfigure}
  \\[3mm]
  \begin{subfigure}{0.22\textwidth}
    \centering
    \scriptsize
    \begin{tikzpicture}[-latex]
      \pnplace{p-p-t-ell}{$p$}[$t$][$\ell$]
      \pntrans{t-p-t-ell}{$p$}[$t$][$\ell$][right=of p-p-t-ell]
      \pnplace{p-p-t-ell-1}{$p$}[$t$][$\ell + 1$][right=of t-p-t-ell]
      \pnplace*{forked}{$p'$}[$t'$][$\RB{entry}_{p'}$][above=of t-p-t-ell]
      \draw (p-p-t-ell) to (t-p-t-ell);
      \draw (t-p-t-ell) to (p-p-t-ell-1);
      \draw (t-p-t-ell) to (forked);
    \end{tikzpicture}
    \caption{$p, t, \ell, r \colon \RB{fork}\:p'\:t'$}
    \label{fig:rb2pn:fork}
  \end{subfigure}
  \begin{subfigure}{0.22\textwidth}
    \centering
    \scriptsize
    \begin{tikzpicture}[-latex]
      \pnplace{p-p-t-ell}{$p$}[$t$][$\ell$]
      \pntrans{t-p-t-ell}{$p$}[$t$][$\ell$][right=of p-p-t-ell]
      \pnplace{p-p-t-ell-1}{$p$}[$t$][$\ell + 1$][right=of t-p-t-ell]
      \pnplace{forked}{$p$}[$t$][$\joins{\ell}$][above=of t-p-t-ell]
      \pntrans{term}{$p'$}[$t'$][$\RB{exit}_{p'}$][above=of p-p-t-ell]
      \draw (p-p-t-ell) to (t-p-t-ell);
      \draw (t-p-t-ell) to (p-p-t-ell-1);
      \draw (term) to (forked);
      \draw (forked) to (t-p-t-ell);
    \end{tikzpicture}
    \caption{$p, t, \ell, r \colon \RB{join}\:t'$}
    \label{fig:rb2pn:join}
  \end{subfigure}
  \begin{subfigure}{0.25\textwidth}
    \centering
    \scriptsize
    \begin{tikzpicture}[-latex,node distance=5mm and 5mm]
      \pnplace{p-p-t-ell}{$p$}[$t$][$\ell$]
      \pntrans{t-p-t-ell}{$p$}[$t$][$\ell$][right=of p-p-t-ell]
      \pnplace{p-p-t-ell-1}{$p$}[$t$][$\ell + 1$][right=of t-p-t-ell]
      \pnplace*{v-t}{$t$}[$v$][][above=5mm of t-p-t-ell]
      \pnplace{v-t2}{}[][$t_2\ v$][below=of p-p-t-ell]
      \pnplace{v-t1}{}[][$t_1\ v$][left=of v-t2]
      \pnplace{v-tn}{}[][$t_n\ v$][below=of p-p-t-ell-1]
      \pnplace{v-tk}{}[][][below=of t-p-t-ell,draw=none]
      \node at (v-tk) {$\cdots$};
      \draw (p-p-t-ell) to (t-p-t-ell);
      \draw (t-p-t-ell) to (p-p-t-ell-1);
      \draw (v-t) to[bend right=35] (t-p-t-ell);
      \draw (t-p-t-ell) to[bend right=35] (v-t);
      \draw (v-t1) to[bend right=10] (t-p-t-ell);
      \draw (v-t2) to (t-p-t-ell);
      \draw (v-tn) to (t-p-t-ell);
      \draw[dotted] (v-tk) to[bend right=20] (t-p-t-ell);
    \end{tikzpicture}
    \caption{$p, t, \ell, r \colon \RB{acquire}\:v$}
    \label{fig:rb2pn:acquire}
  \end{subfigure}
  \begin{subfigure}{0.25\textwidth}
    \centering
    \scriptsize
    \begin{tikzpicture}[-latex,node distance=5mm and 5mm]
      \pnplace{p-p-t-ell}{$p$}[$t$][$\ell$]
      \pntrans{t-p-t-ell}{$p$}[$t$][$\ell$][right=of p-p-t-ell]
      \pnplace{p-p-t-ell-1}{$p$}[$t$][$\ell + 1$][right=of t-p-t-ell]
      \pnplace*{v-t}{$t$}[$v$][][above=5mm of t-p-t-ell]
      \pnplace{v-t2}{}[][$t_2\ v$][below=of p-p-t-ell]
      \pnplace{v-t1}{}[][$t_1\ v$][left=of v-t2]
      \pnplace{v-tn}{}[][$t_n\ v$][below=of p-p-t-ell-1]
      \pnplace{v-tk}{}[][][below=of t-p-t-ell,draw=none]
      \node at (v-tk) {$\cdots$};
      \draw (p-p-t-ell) to (t-p-t-ell);
      \draw (t-p-t-ell) to (p-p-t-ell-1);
      \draw (t-p-t-ell) to[bend left=35] (v-t);
      \draw (v-t) to[bend left=35] (t-p-t-ell);
      \draw (t-p-t-ell) to[bend left=10] (v-t1);
      \draw (t-p-t-ell) to (v-t2);
      \draw (t-p-t-ell) to (v-tn);
      \draw[dotted] (t-p-t-ell) to[bend left=20] (v-tk);
    \end{tikzpicture}
    \caption{$p, t, \ell, r \colon \RB{release}\:v$}
    \label{fig:rb2pn:release}
  \end{subfigure}
  \\[3mm]
  \begin{subfigure}{0.8\textwidth}
    \centering
    \scriptsize
    \begin{tikzpicture}[-latex]
      \coordinate (p-leftmost);
      \coordinate [right=of p-leftmost] (t-leftmost);
      \pnplace{p-p-t-ell-1}{$p$}[$t$][$\ell\: r_1$][above=3.5mm of p-leftmost]
      \pntrans{t-p-t-ell-1}{$p$}[$t$][$\ell\: r_1$][right=of p-p-t-ell-1]
      \pnplace{p-p-t-ell-2}{$p$}[$t$][$\ell\: r_2$][below=3.5mm of p-leftmost]
      \pntrans{t-p-t-ell-2}{$p$}[$t$][$\ell\: r_2$][right=of p-p-t-ell-2]
      \pnplace{p-entry}{$p'$}[$t$][$\RB{entry}_{p'}\: \ell$][right=20mm of t-leftmost]
      \node[right=of p-entry] (ddd) {$\cdots$};
      \pntrans{t-exit}{$p'$}[$t$][$\RB{exit}_{p'}\: \ell$][right=of ddd]
      \pnplace{p-ret}{$p'$}[$t$][$\rets{\ell}$][right=12mm of t-exit]
      \coordinate[right=9mm of p-ret] (t-ret);
      \pntrans{t-ret-1}{$p$}[$t$][$\rets{\ell}\: r_1$][above=3.5mm of t-ret]
      \pnplace{p-ret-ell-1-1}{$p$}[$t$][$\ell + 1\: r_1$][right=8mm of t-ret-1]
      \pntrans{t-ret-2}{$p$}[$t$][$\rets{\ell}\: r_2$][below=3.5mm of t-ret]
      \pnplace{p-ret-ell-1-2}{$p$}[$t$][$\ell + 1\: r_2$][right=8mm of t-ret-2]
      
      \node [fit=(p-entry)(t-exit),cloud,draw,cloud puffs=30,cloud puff arc=120, aspect=3, inner xsep=3mm] {};

      \draw (p-p-t-ell-1) to (t-p-t-ell-1);
      \draw (p-p-t-ell-2) to (t-p-t-ell-2);
      \draw (t-p-t-ell-1) to (p-entry);
      \draw (t-p-t-ell-2) to (p-entry);
      \draw (t-exit) to (p-ret);
      \draw (p-ret) to (t-ret-1);
      \draw (p-ret) to (t-ret-2);
      \draw (t-ret-1) to (p-ret-ell-1-1);
      \draw (t-ret-2) to (p-ret-ell-1-2);
    \end{tikzpicture}
    \caption{$p, t, \ell, r_{1, 2} \colon \RB{call}\:p'$}
    \label{fig:rb2pn:call}
  \end{subfigure}
  \caption{Fragments of Petri nets encoding \rb commands.}
  \label{fig:rb-to-pn}
\end{figure}

\subsection{Encoding of \rb into Petri Nets}
\label{sec:rb-to-pn}

Given an \rb program $R$ with threads $T$, procedures $P$, labels $B$, and
variables $V$, \tech builds a PN $N = \encoding(R) = (\Pi, \Delta, A, I)$ that over-approximates $R$'s semantics.
The general idea is that
when the command at location $\ell$
in procedure $p$ is ready to be executed by thread $t$,
the corresponding place $(p, t, \ell, r)$ is marked.
The component $r$ denotes the \emph{caller} location of $p$,
which is $\top$ if $p$ is being executed on the main thread $t_0$
or on a freshly forked thread.
Then, the corresponding transition $(p, t, \ell, r)$ fires
when the command executes;
the transition's postset denotes
the commands that will be able to execute after $c$.
With this approach, each combination of
procedure, thread, and call site in $R$
corresponds to a set of connected nodes (places and transitions) in PN $N$.

Formally, 
the set $\Pi$ of \emph{places} of PN $N$ includes, 
for every combination of
location $\ell \in B$, procedure $p \in P$, thread $t \in T$, and return label $r \in B \cup \{ \top \}$:
\begin{enumerate*}
\item A place $(p, t, \ell, r)$
  if $\ell$ is a location in $p$ and one of the following holds:
  \begin{enumerate*}[label=\emph{\alph*})]
  \item $p$ is \RB{main}, $t$ is $t_0$, and $r$ is $\top$; or
  \item there is a $\RB{fork}\:p\:t$ somewhere in the program, and $r$ is $\top$; or
  \item there is a $\RB{call}\:p$ at location $r \neq \top$ in the program.
  \end{enumerate*}
\item A return place $(p, t, \rets{\ell})$
  if $\ell$ is the location of a $\RB{call}\:p$ command.
\item A join place $(p, t, \joins{\ell}, r)$
  if $\ell$ is the location of a $\RB{join}\:t'$ command.
\item A place $(t, v)$ for every variable $v \in V$ such that
  there is an $\RB{acquire}\ v$ or a $\RB{release}\ v$ somewhere in the program.
\end{enumerate*}
The set $\Delta$ of \emph{transitions} includes,
for every combination of
location $\ell \in B$, procedure $p \in P$, thread $t \in T$, and return label $r \in B \cup \{ \top \}$:
\begin{enumerate*}
\item A transition $(p, t, \ell, r)$
  for every place $(p, t, \ell, r) \in \Pi$.
\item A return transition $(p, t, \rets{\ell}, r)$
  if $\ell$ is the location of a $\RB{call}\:p$ command;
\item Two jump transitions $(p, t, \ell\!\to\!\ell_1, r)$ and
  $(p, t, \ell\!\to\!\ell_2, r)$ if $\ell$ is the location of a $\RB{jump}\:\ell_1\:\ell_2$ command.
\end{enumerate*}
The \emph{initial marking} $I$ has one token in the place $(\RB{main}, t_0, \RB{entry}_{\RB{main}}, \top)$---the program's unique entry point---and one token in each place $(t, v)$---denoting that all locks are initially not held by any thread.

For each command $c$ at location $\ell$ in procedure $p$
executed by thread $t$ with return location $r$,
\autoref{fig:rb-to-pn} shows the set $A$ of \emph{arcs} that capture the command's semantics.
For readability, \autoref{fig:rb-to-pn} omits the caller location $r$
for commands where it's the same in all nodes (i.e., all commands but calls).
For brevity, we only describe the most complex fragments:
\begin{enumerate*}[label=\emph{\alph*})]
\item Command $\RB{fork}\:p'\:t'$ puts one token into the next location's place,
  and one token into the forked thread's entry point's place $(p', t', \RB{entry}_{p'}, r)$.
  This way, the forked thread's computation can proceed in parallel to the forking thread's.
\item Conversely, $\RB{join}\ t'$ can fire only when
  the exit transition $(p', t', \RB{exit}_{p'}, r)$ of thread $t'$
  fires (where $p'$ is the procedure that thread $t'$ is running).
\item Command $\RB{acquire}\:v$'s transition can fire only if \emph{all}
  places $(t', v)$ are marked, denoting that no thread $t'$ holds a lock on $v$.
  When it fires, it puts back a token only in place $(t, v)$
  to indicate that $t$ holds a lock on $v$.
\item Command $\RB{release}\:v$'s transition can fire only if place $(t, v)$
  is marked, denoting that thread $t$ holds a lock on $v$.
  When it fires, it puts back a token in all places $(t', v)$,
  thus allowing other threads to acquire a lock on $v$.
\item When the transition $(p, t, \ell, r)$ of command $\RB{call}\:p'$
  fires, it puts a token in the callee's entry place $(p', t, \RB{entry}_{p'}, \ell)$.
  Then, execution continues in the subnet corresponding procedure $p'$ called in thread $t$
  at call site $\ell$.
  Then, when transition $(p', t, \RB{exit}_{p'}, \ell)$ fires,
  signaling that the subnet's execution has terminated,
  a token first goes into return place $(p', t, \rets{\ell})$;
  then one of the transitions $(p, t, \rets{\ell}, r)$, for all possible callers of $p'$,
  nondeterministically fires.
  \autoref{fig:rb2pn:call} pictures two callers $r_1 \neq r_2$,
  where either of the return transitions $(p, t, \rets{\ell}, r_k)$, $k = 1, 2$,
  may fire when the subnet terminates execution.
  This nondeterministic encoding is the only aspect
  of \tech's encoding of \rb programs that loses precision:
  a PN cannot store an unbounded stack of return locations,\footnote{
    PNs with inhibitor arcs would be able to simulate this
    without loss of precision; however, their reachability problem becomes undecidable~\cite{FMMR-TimeBook-12}.
  }
  and hence \tech overapproximates it.
\end{enumerate*}

\begin{expara}
  \autoref{fig:ex:pn} shows the PN encoding of \autoref{fig:ex:rb}
  and its initial marking.
  For readability, \RB{m} abbreviates \RB{main},
  and we omit the caller in nodes since it's $\top$ everywhere.
  The PN consists of three subnets, corresponding
  to the \RB{main} (middle), \RB{r_12} (top), and \RB{r_21} (bottom)
  procedures; as well as four additional places $(\RB{t_}xy, \RB{lock}m)$
  that denote when $\RB{t_}xy$ is holding a lock on $\RB{lock}m$.
  The arcs connecting \RB{main} to the other subnets
  mark the spawning of each thread in the main thread $t_0$;
  and the colored arcs connecting transitions
  in the \RB{r_12} and \RB{r_21} subnets to
  the places $(\RB{t_}xy, \RB{lock}m)$ with matching color
  synchronize the corresponding lock commands.
\end{expara}

\begin{lemma}[Correctness of Petri net encoding] \label{lemma:pn-correct}
  Let $P$ be an \rb program and $N = \encoding(P)$ its PN encoding defined by \tech.
  Then, each sequence of state transitions according to $P$'s operational
  semantics corresponds to a marking sequences according to $N$'s semantics.
  Therefore, the semantics of $N$ is a sound over-approximation of $P$'s semantics.
\end{lemma}
\begin{proof}
  See appendix.
\end{proof}

\subsection{Encoding Properties and Implementation}
\label{sec:encoding-props}

The trace semantics of the PN $N = \encoding(\translation(J))$
built by \tech
is a path-insensitive approximation of the trace semantics of
the input bytecode program $J$.
Thus, $N$ can be analyzed with any standard PN analyzer,
such as LoLA~\cite{LoLA}.
While this approach gives flexibility,
since one may formalize a wide range of properties as
formulas in the temporal logic supported by the analyzers,
doing so requires some knowledge of how \tech
encodes programs into PNs.
To improve \tech's usability,
we built in support for common concurrency bugs:
\emph{deadlock}, \emph{livelock}, and \emph{non-termination}.

We implemented the \tech technique in a command-line tool called \tool.
\tool inputs a bytecode program $J$, a property to be verified
(a temporal logic formula, or absence of deadlocks, livelocks, or termination),
uses Soot to analyze Jimple code and its control-flow,
and produces a PN $P = \encoding(\translation(J))$ as described in previous sections.

\section{Experimental Evaluation}
\label{sec:experiments}

\nicepar{Evaluation goal.}
The goal of this experimental evaluation is to assess
\tech's practical feasibility, highlighting
its capabilities, limitations, and complementarity
in comparison to other automated Java verification tools that can analyze similar concurrency properties.

\subsection{Subjects and Setup}
\label{sec:subjects}

\nicepar{Comparable tools.\;}
According to \autoref{sec:related-work}'s discussion,
we focus on JPF (Java Pathfinder~\cite{JPF}) and JaDA~\cite{JaDA} for the following reasons:
\begin{enumerate*}
\item they support checking (absence of) deadlocks out of the box---the same properties currently supported by \tool;
\item their public repositories (especially JaDA's) include numerous examples of Java programs that demonstrate their capabilities in deadlock detection;
\item they represent two distinct approaches: JPF is based on model-checking, on top of which it implements an array of analysis techniques, whereas JaDA uses type-based static analysis and is specialized on deadlock detection;
\item their implementations have different levels of maturity: JPF is a mature tool with over 20 years of history,
  whereas JaDA is more recent and less polished, but with a ready-made, usable implementation.
\end{enumerate*}
Our experiments are not meant as a direct
\emph{comparison} between \tool and these two tools;
however, they will
show that \tool works successfully on several examples that were designed for JPF and JaDA,
indicating that our approach is viable and has potential, as well as its limitations compared to the state of the art.

\nicepar{Subjects.}
We evaluated \tool's capabilities on 39 example concurrent programs
with shared locks.
The leftmost part of \autoref{tab:results:subjects} lists these programs,
which belong to 7 groups:
\begin{enumerate*}
\item 2 are from \textsl{JPF}'s repository;
\item 14 are from \textsl{JaDA}'s repository~\cite{JaDA-examples};
\item 2 are from the JaConTeBe~\cite{JaConTeBe} 
  \textsl{benchmarks};
\item 7 are \textsl{variants} of other examples:
  2 are variants of programs in group \textsl{JaDA}
  (\J{ClassicDeadlockN} is a generalization of \J{ClassicDeadlock},
  and \J{PhilTable} is a dining philosophers variant),
  and the other 5 use recent Java language features (up to Java 21);
\item 11 are examples that specifically scale up the \textsl{size} and complexity of the verification problems by introducing nested loops, several threads, sequences of blocks, and unbounded recursion;
\item 3 are adaptations of other classic examples translated to \textsl{Kotlin}.
\end{enumerate*}
We selected deadlock examples from JPF's and JaDA's repositories that don't use features unsupported by \tool (especially arrays, which result in a major loss of precision); for the same reason, we refactored some of the JaDA examples to make them compatible with \tool without changing their behavior. 
As indicated in \autoref{tab:results},
3 of the subjects are correct, and 36 include a concurrency (deadlock) bug.
Therefore, the evaluation tests the tools' capabilities both to precisely find bugs
and to soundly verify correctness.

\nicepar{Challenges.}
Although most benchmark programs
are relatively small in terms of lines, methods, or classes,
they are far from trivial, and 
introduce various kinds of challenges to even state-of-the-art tools.
In particular:
\begin{enumerate*}
\item Group \textsl{variants}
  uses language features (sealed classes, records, etc.)
  that have only been introduced in recent versions of Java;
  for approaches that work on the source level (like JPF and JaDA),
  supporting new language features is a major challenge
  that often goes beyond ``mere'' engineering.
  Notably, program \J{VirtualThreads} specifically targets a
  new \emph{concurrency} feature of Java that \tool supports out of the box.
  On the other hand, program \J{PhilTable}
  uses a lot of aliased references,
  which challenge the tools' abstraction of this feature.
\item Programs in group \emph{size}
  feature different forms of complex control flow,
  which challenges the \emph{scalability} of concurrency analysis:
  unbounded recursive synchronous calls (\J{RecUnbounded}, \J{Chordv2});
  recursive thread spawning and joining (\J{RecJoin});
  thread synchronization with unbounded loops
  (\J{WhileCnt2N2T}),
  deeply nested calls (\J{Interleaving}),
  and nested locking of many lock variables (\J{Nested200T2});
  and an 8-thread generalization of the dining philosophers problem (\J{DiningP8T}).
  The remaining programs in this group
  feature behaviors that challenge 
  concurrency analysis
  with complex \emph{dependency} patterns (\J{InnerThreads},
  \J{JoinUnderLock}, \J{RecursiveJoinUnderLockDeadlock}),
  such as circular waiting between parent and child threads,
  or require a context-sensitive analysis (\J{DeeperContextDeadlock}).
\end{enumerate*}

\nicepar{Setup.}
We ran each example with \tool, JPF, and JaDA
with 15-minute timeout;
we classify a tool's output as:
\begin{enumerate*}
\item success \success\ if it correctly detects a deadlock or
  establishes that there is none;
\item false negative \falseneg\ if it misses an existing deadlock;
\item false positive \falsepos\ if it reports a spurious bug;
\item time out/out of memory \timeout;
\item failure due to unsupported features \unsupp.
\end{enumerate*}
All experiments ran on an Apple Macbook M3 Max with 36~GB RAM and macOS~15.6.1.
\tool used Soot~4.6.0 and LoLA~2.0;
JPF commit~\texttt{0f2f2} used Java~11;
the JaDA tool was run through its website\footnote{\url{http://jada.cs.unibo.it/}}
since it is not available for download.

\begin{table}[!tb]
  \centering
\setlength{\tabcolsep}{2.4pt}
  \renewcommand{\arraystretch}{0.5}
  \scriptsize
  \begin{tabular}{crrrrr|rrrr|ccc}
    \toprule
    \multicolumn{1}{c}{\textsc{group}} & \multicolumn{1}{c}{\textsc{name}} & \multicolumn{1}{c}{\textsc{exp}} & \multicolumn{1}{c}{\textsc{loc}} & \multicolumn{1}{c}{\textsc{classes}} & \multicolumn{1}{c}{\textsc{methods}} & \multicolumn{1}{|c}{$|\textrm{PN}|$} & \multicolumn{1}{c}{\textsc{soot}} & \multicolumn{1}{c}{\textsc{enc}} & \multicolumn{1}{c}{\textsc{ver}} & \multicolumn{3}{c}{\textsc{outcome}} \\
    &&&&&&&[s]&[s]&[s]& {\scriptsize \textsc{jPe}} & {\scriptsize \textsc{JPF}} & {\scriptsize \textsc{JaDA}} \\
    \midrule
\multirow{2}{*}{\textsl{JPF}} & \texttt{BankTransfer} & \isBug & 18 & 2 & 4 & 91 & 175.9 & 0.5 & 0.1 & \success & \success & \unsupp \\
 & \texttt{DiningPhil} & \isBug & 40 & 3 & 3 & 133 & 212.1 & 1.5 & 0.5 & \success & \success & \falseneg \\
\cmidrule(lr){2-13}
\multirow{14}{*}{\textsl{JaDA}} & \texttt{AnythingTest} & \isBug & 36 & 2 & 5 & 77 & 172.4 & 0.4 & 0.2 & \success & \success & \success \\
 & \texttt{BuildNetwork} & \isBug & 52 & 3 & 4 & 112 & 157.1 & 0.5 & 0.1 & \success & \success & \falseneg \\
 & \texttt{Chord} & \isBug & 37 & 1 & 7 & 67 & 306.8 & 2.5 & 0.9 & \success & \success & \success \\
 & \texttt{ClassicDeadlock} & \isBug & 48 & 1 & 6 & 73 & 148.4 & 0.4 & 0.1 & \success & \success & \success \\
 & \texttt{ClassicPhil} & \isBug & 51 & 2 & 8 & 217 & 148.6 & 0.6 & 0.2 & \success & \timeout & \success \\
 & \texttt{DanglingThreads} & \isBug & 59 & 1 & 10 & 91 & 150.4 & 0.5 & 0.1 & \success & \success & \success \\
 & \texttt{DeadlockTwo} & \isBug & 34 & 1 & 5 & 112 & 149.8 & 0.5 & 0.1 & \success & \success & \falseneg \\
 & \texttt{GuardedLocks} & \isCorrect & 30 & 1 & 3 & 64 & 165.3 & 0.5 & 0.1 & \success & \success & \falsepos \\
 & \texttt{MayNotHappenInParallel} & \isCorrect & 26 & 1 & 3 & 53 & 166.1 & 0.4 & 0.2 & \success & \success & \success \\
 & \texttt{NetworkAllP} & \isBug & 52 & 3 & 4 & 112 & 147.8 & 0.4 & 0.1 & \success & \success & \falseneg \\
 & \texttt{PhilTableP} & \isBug & 60 & 4 & 5 & 203 & 331.1 & 3.7 & 0.8 & \success & \timeout & \falseneg \\
 & \texttt{SimpleWhile} & \isBug & 31 & 2 & 3 & 80 & 155.0 & 0.4 & 0.1 & \success & \timeout & \success \\
 & \texttt{StaticFields} & \isBug & 28 & 1 & 3 & 65 & 167.1 & 0.5 & 0.1 & \success & \success & \falseneg \\
 & \texttt{SynchMethod} & \isBug & 21 & 1 & 3 & 46 & 169.3 & 0.5 & 0.1 & \success & \success & \success \\
\cmidrule(lr){2-13}
\multirow{2}{*}{\textsl{benchmarks}} & \texttt{Deadlock1} & \isBug & 40 & 4 & 3 & 72 & 148.7 & 0.4 & 0.1 & \success & \success & \falseneg \\
 & \texttt{Deadlock2} & \isBug & 25 & 2 & 3 & 51 & 148.5 & 0.3 & 0.1 & \success & \success & \success \\
\cmidrule(lr){2-13}
\multirow{7}{*}{\textsl{variants}} & \texttt{ClassicDeadlockN} & \isCorrect & 42 & 1 & 5 & 73 & 322.4 & 2.9 & 0.8 & \success & \success & \success \\
 & \texttt{PhilTable} & \isBug & 59 & 2 & 9 & 234 & 162.0 & 88.9 & 214.5 & \timeout & \timeout & \success \\
 & \texttt{Records} & \isBug & 55 & 4 & 7 & 71 & 314.9 & 3.2 & 0.9 & \success & \unsupp & \unsupp \\
 & \texttt{Sealedclasses} & \isBug & 57 & 6 & 7 & 71 & 313.4 & 2.9 & 0.7 & \success & \unsupp & \unsupp \\
 & \texttt{Switch} & \isBug & 59 & 4 & 7 & 91 & 314.5 & 4.1 & 1.0 & \success & \unsupp & \unsupp \\
 & \texttt{VirtualThreads} & \isBug & 57 & 3 & 5 & 12 & 157.9 & 0.68 & 0.2 & \success & \unsupp & \unsupp \\
 & \texttt{Yield} & \isBug & 50 & 3 & 6 & 73 & 317.0 & 2.8 & 0.7 & \success & \unsupp & \unsupp \\
\cmidrule(lr){2-13}
\multirow{11}{*}{\textsl{size}} & \texttt{Chordv2} & \isBug & 37 & 1 & 7 & 47 & 306.3 & 2.8 & 0.7 & \success & \timeout & \falseneg \\
 & \texttt{DeeperContextDeadlock} & \isBug & 65 & 1 & 8 & 186 & 303.8 & 4.1 & 0.7 & \success & \success & \falseneg \\
 & \texttt{DiningPhilP8T} & \isBug & 60 & 3 & 19 & 666 & 276.6 & 4.9 & 0.9 & \success & \timeout & \falseneg \\
 & \texttt{InnerThreads} & \isBug & 32 & 1 & 4 & 54 & 308.8 & 2.4 & 0.5 & \success & \success & \success \\
 & \texttt{Interleaving} & \isBug & 155 & 1 & 30 & 547 & 169.1 & 1.1 & 0.3 & \success & \timeout & \falseneg \\
 & \texttt{JoinUnderLock} & \isBug & 22 & 1 & 3 & 40 & 307.7 & 2.4 & 0.7 & \success & \success & \falseneg \\
 & \texttt{Nested200T2} & \isBug & 1015 & 1 & 6 & 1870 & 271.0 & 11.8 & 4.1 & \success & \success & \timeout \\
 & \texttt{RecJoin} & \isBug & 44 & 1 & 6 & 175 & 284.9 & 2.8 & 0.7 & \success & \success & \timeout \\
 & \texttt{RecUnbounded} & \isBug & 69 & 3 & 6 & 274 & 297.6 & 3.4 & 0.7 & \success & \success & \falseneg \\
 & \texttt{RecursiveJoinUnderLock} & \isBug & 35 & 1 & 5 & 89 & 298.4 & 2.9 & 0.7 & \success & \success & \falseneg \\
 & \texttt{WhileCnt2N2T} & \isBug & 34 & 1 & 5 & 91 & 282.0 & 3.1 & 0.6 & \success & \timeout & \falseneg \\
\cmidrule(lr){2-13}
\multirow{3}{*}{\textsl{Kotlin}} & \texttt{Chord} & \isBug & 41 & 1 & 6 & 88 & 148.3 & 0.5 & 0.1 & \success & \success & \unsupp \\
 & \texttt{ClassicDeadlock} & \isBug & 29 & 1 & 2 & 64 & 147.9 & 0.5 & 0.1 & \success & \success & \unsupp \\
 & \texttt{When} & \isBug & 45 & 1 & 4 & 89 & 148.4 & 0.5 & 0.1 & \success & \success & \unsupp \\
\midrule
\textbf{total} &  &  & 2738 & 76 & 238 & 6624 & 8818.8 & 164.2 & 234.2 &  &  &  \\
\textbf{average} &  &  & 70 & 2 & 6 & 170 & 226.1 & 4.2 & 6.0 &  &  & 

\\ \bottomrule
  \end{tabular}
\caption{Experimental subjects and results.
For each \textsc{group}
    of programs used to evaluate \tool, the table lists the \textsc{name} of each program in the group,
    whether the program is correct \isCorrect\ has a deadlock \isBug,
    its size in lines of code \textsc{loc}, number of \textsc{classes} and \textsc{methods}.
    The table also reports the size $|\textrm{PN}|$
    in number of places of \tool's PN model of the program,
    the times (in seconds) for
    \textsc{soot}'s analysis, \tool's PN \textsc{enc}oding,
    and LoLA's \textsc{ver}ification.
    The right-hand side details the \textsc{outcome} of verification with each tool: \tool, JPF, and JaDA:
    success \success, false negative \falseneg,
    false positive \falsepos,
    time out/out of memory \timeout,
    and failure due to unsupported features \unsupp.
    }
  \label{tab:results}
  \label{tab:results:subjects}
\end{table}

\subsection{Results}
\label{sec:results}

\paragraph{Performance.}
\autoref{tab:results:subjects} shows that \tool's encoding is 
succinct,
with 2.4 (6624/2738) places per source line of code on average.
Also as a result of this succinctness,
verification is fast in most examples,
typically taking under a second.
In contrast, generating the PN encoding from bytecode
takes significantly more time;
however,
most of the generation time (3.8 \emph{minutes} per subject) is taken by Soot
to extract from bytecode
the information needed for the \rb encoding,
whereas \tool's actual encoding is 
much faster
(4.2 seconds per subject).
The overall end-to-end performance of \tool is acceptable,
but there are clear margins for improvement,
in particular as
we invest more time into optimizing its bottleneck interaction with Soot.
It is also encouraging that the two outliers
(\J{Nested200T2} and especially \J{PhilTable})
that took considerably more time
resulted in larger-than-average PN models that similarly
challenged encoding and verification;
in other words, they are intrinsically more complex examples.
We did not collect the running time of JPF or JaDA:
a proper comparison of performance would require controlled experimental conditions
(not possible for JaDA whose implementation is not open source),
and would be of limited interest for a prototype tool like \tool.

\tool correctly verified
(detecting a deadlock or confirming that none exists)
38 examples.
The lone exception is \J{PhilTable},
where Soot's alias analysis induces an unnecessarily large
number of lock-thread combinations,
which results in a huge state space
on which \tool runs out of memory.
This example showcases a current limitation of \tool,
which depends on the precision of Soot's alias analysis.
The flip side is that plugging in a bespoke alias analysis---something for future work---in \tool's implementation will improve its precision without
changes to the underlying technique.

\paragraph{Language feature support.}
Since \tool targets bytecode,
it can verify programs in group \textsl{variants} that use
features that have only been available in recent Java versions.
In contrast, both JPF and JaDA
fail to analyze 5 examples in \textsl{variants}
since they do not support recent language features;
for the same reason,
JaDA also fails on the \J{BankTransfer} example that uses lambdas.
\tool can also analyze the 3 examples written in \textsl{Kotlin};
somewhat surprisingly,
JPF can also analyze them,
but only
because they use basic features that compile to a small bytecode instruction set
that JPF can input.
Instead, JaDA is strictly limited to programs written in Java.
As explained above, we deliberately selected these examples
to demonstrate \tool's adaptability to recent Java features and other JVM languages;
\tool has other language limitations,
which we discussed in \autoref{sec:approach}.

\paragraph{Scalability.}
\tool successfully analyzed 4 examples
in group \textsl{size}
and 3 examples in group \textsl{JaDA}
on which JPF ran out of memory.
This indicate that JPF
struggles to scale
to programs
with significant usage of recursion
(\J{ClassicPhil}, \J{DiningPhilP8T}, \J{Chordv2}, \J{PhilTableP})
or a large number of interleavings
(\J{SimpleWhile}, \J{WhileCnt2N2T}, \J{Interleaving}).
Interestingly, JPF times out too on the one example \J{PhilTable}
where \tool runs out of memory.
While JaDA's abstractions make it a generally more scalable tool,
it still timed out on 2 examples in group \textsl{size} that
\tool can analyze.
In all, these examples indicate that the kinds of complex behaviors
that \tool can analyze are somewhat
complementary to the focus of tools based on different approaches.

\paragraph{Soundness and precision.}
In all our experiments neither \tool nor JPF produced
any false positives or false negatives.
While \tech relies on approximations, which we presented formally in the paper,
that may break soundness or precision,
these experiments indicate that \tool remains practically applicable
on numerous examples despite its theoretical limitations.
In contrast,
JaDA incurred 15 false negatives and 1 false positive in the experiments.
The false negatives indicate
that its abstractions are not sound in general,
since they sometimes omit program interactions
that are feasible and may trigger a deadlock.
The false positive (\J{GuardedLocks})
is simply a result of an imprecise overapproximation
of the possible sequential executions under locking.

\paragraph{Evaluation summary.}
Our evaluation indicates that
\tool is applicable to programs using modern Java features,
including concurrency features such as virtual threads.
Despite being a prototype, \tool can tackle
programs with complex control flow (nested conditionals, loops, recursion)
and dynamic, unbounded thread creation.
Its main limitations are the lack of support for
Java features like exceptions and arrays, and its
imprecision in the presence of complex aliasing.
Overall, its current capabilities and limitations
complement other Java concurrency analysis tools.

\paragraph{Limitations.}
Our experiments are are not meant to be a direct comparison between \tool and JPF and JaDA,
both because JaDA and, especially, JPF are much more mature tools,
and because their characteristics are largely complementary.

\section{Conclusions}
This paper presented \tech: a technique to analyze concurrency properties
of Java programs based on model-checking an encoding of a flow- and context-sensitive, path-insensitive approximation of the bytecode semantics.
Two characteristics distinguish \tech's approach from other automated verification techniques for concurrent Java:
\begin{enumerate*}
\item \tech uses Petri nets (PNs) to approximate the program's concurrent behavior and to apply model checking.
\iflong
  Compared to other state/transition models, PNs model concurrency naturally and succinctly, and feature an interesting
  expressive power that goes beyond finite-state models but is still amenable to automated verification.
\fi
\item \tech operates on Java bytecode rather than source code. This 
  brings several practical advantages, in particular in terms of robustness and flexibility. \iflong In fact,
  \tech's implementation (called \tool) works on any version of Java, whereas most similar verification tools
  are limited to (usually a subset of) language features introduced in older versions of Java.
  \tech can even verify programs written in a subset of Kotlin---another language that compiles to Java bytecode.
  \fi
\end{enumerate*}
The experiments we discussed in the paper demonstrate the practical advantages brought by \tech's design choice,
as well as its current limitations.
Even though \tool is still a prototype tool,
we could find several examples of programs that challenge other more mature Java automated concurrency verification tools
that it can analyze successfully.

\iflong
\iflong
In future work, we plan to further explore the capabilities of PNs for concurrency verification mainly in three directions.
First, we will consider PN variants with different features (e.g., inhibitor arcs, colors, etc.)
and explore their expressiveness vs.\ scalability trade-offs.
Second, we will extend \tech's capabilities by modeling other kinds of concurrency properties.
Third, we will improve \tech's precision and soundness by developing a bespoke alias analysis (and other static analyses)
that better fulfills the requirements of concurrency analysis of Java programs.
\else
  In future work, we plan to further explore the capabilities of PNs for concurrency verification by
  considering PN variants with different features
  (e.g., inhibitor arcs, colors, etc.)
and exploring their expressiveness vs.\ scalability trade-offs.
\fi
\fi

\clearpage



\iflong

\clearpage
\appendix

\section{Additional details for \autoref{sec:jimple}}

A \nt{Jimple} program is a collection of methods,
with at least one \Ji{main} method
that correspond to the program's entry point.
A \nt{Method} has a name, a typed signature, and a body
consisting of a sequence of instructions:
\begin{features}
\item[Control flow] instructions include different kinds of jumps:
  \begin{enumerate}
  \item \Ji{goto label}:
    unconditional jump to the location with a given \Ji{label};
  \item \Ji{if $\:c$ label}:
    jump to \Ji{label} if condition $c$ is true;
  \item $\Ji{switch}\ v\ (v_1 \colon \Ji{label}_1) (v_2 \colon \Ji{label}_2)\ldots$:
    if variable $v$ is equal to any $v_k$ jump to the corresponding $\Ji{label}_k$.
    For simplicity, we present a single \Ji{switch} instruction,
    even though Jimple includes two variants \Ji{lookupswitch} and \Ji{tableswitch}
    whose differences are unimportant for this work.
  \item $\Ji{return}\:v$ and \Ji{return}:
    return to the caller (possibly with return value $v$).
  \end{enumerate}

\item[Monitor] instructions correspond to the synchronization that occurs when entering a \J{synchronized} block in Java:
  \begin{enumerate}[resume]
  \item \Ji{monitor_enter $\:v$}: acquire a lock on $v$; 
  \item \Ji{monitor_exit $\:v$}: release the lock on $v$.
  \end{enumerate}

\item[Assignment] instructions perform any kind of expression evaluation
  using the SSA form discussed above:
  \begin{enumerate}[resume]
  \item $t\:\Ji{:=}\:v_1 \oplus v_2$:
    evaluate binary operation $\oplus$ with arguments $v_1$ and $v_2$,
    and store the result in variable~$t$;
  \item $t\:\Ji{:=} \:\odot v$:
    evaluate unary operation $\odot$ with argument $v$,
    and store the result in variable $t$;
  \item $t\:\Ji{:=}\:c$:
    store constant value $c$ in variable $t$.
  \end{enumerate}

\item[Call] instructions are also only used in the right-hand side
  of an assignment; if a call is used as an instruction or returns \J{void},
  we will write $\epsilon := \Ji{invoke}\:m$, where $\epsilon$ denotes
  a dummy local variable.
  Precisely, JVM bytecode includes five variants of call instructions:
  \Bc{static} for \J{static} methods,
  \Bc{virtual} for instance methods,
  \Bc{interface} for abstract interface calls,
  \Bc{special} for constructors and \J{super} calls,
  \Bc{special} for constructors and \J{super} calls,
  and \Bc{dynamic} for lambdas.
  For simplicity, \autoref{fig:syntax-jimple}
  presents a single form $\Ji{invoke}\ m\: v_0\: v_1\ldots$
  to denote a call of method $m$ with given actual arguments $v_0$, $v_1$, \ldots;
  however, we will distinguish between the various kinds of calls
  when relevant to the encoding of Jimple into \rb.
  
\item[Other] instructions:
  \begin{enumerate}[resume]
  \item \Ji{noop}: do nothing.
  \end{enumerate}  
\end{features}

\section{Additional details for \autoref{sec:rb}}

\rb commands include:
\begin{features}
\item[Control flow] commands:
  \begin{enumerate*}
  \item \RB{goto label}:
    unconditional jump to the location with a given \Ji{label};
  \item \RB{jump label$_1$ label$_2$}:
    nondeterministic jump to either $\RB{label}_1$ or $\RB{label}_2$;
  \item \RB{fork $\:p\:t$}:
    spawn a parallel thread $t$ running procedure $p$.
  \item \RB{join $\:t$}:
    wait for thread $t$ to terminate.
  \end{enumerate*}
  We assume that the thread identifiers in \RB{fork} and \RB{join} commands
  are distinct from all other identifiers,
  and that the \RB{main} procedure runs on thread $t_0$.

\item[Synchronization] commands:
  \begin{enumerate*}[resume]
  \item \RB{acquire $\:v$}: acquire a lock on $v$;
  \item \RB{release $\:v$}: release the lock on $v$.
  \end{enumerate*}

\item[Access] commands:
  \begin{enumerate*}[resume]
  \item \RB{read $\:v$}: read variable $v$;
  \item \RB{write $\:v$}: write variable $v$.
  \end{enumerate*}

\item[Call] command
  \begin{enumerate*}[resume]
    \item \RB{call $\:p$}, which executes procedure $p$ and returns
  to the caller without returning any value.
  \end{enumerate*}
  
\item[Other] commands:
  \begin{enumerate*}[resume]
  \item \RB{skip}: do nothing.
  \end{enumerate*}
\end{features}  

\section{Additional details for~\autoref{sec:rb-soundness-precision}}

\begin{description}
\item[Conditionals:] Translating conditionals with
  nondeterministic \RB{jump}s is sound but path-insensitive; hence, it
  generally involves a loss of \emph{precision}, because $R$ may
  include paths that are unfeasible in $J$ due
  to unsatisfiable path conditions.\\
  For example, the \Ji{error} location is unreachable in \autoref{fig:soundness-examples-conditionals}'s Jimple program (left)
  because the condition of the \Ji{if} is identically false;
  however, it becomes reachable in the \rb translation (right)
  where the condition is abstracted away.

\item[Reentrant locks:] The translation of lock operations in $J$ with \RB{acquire} and \RB{release}
  in $R$ is sound provided the lock is not used \emph{reentrantly}.
  According to \autoref{fig:rb-semantics}, \rb does not allow a thread
  to acquire a lock on $k$ if it already holds a lock on it;
  thus, $R$ may omit such executions even if they are possible in $J$.\footnote{
    Soundly modeling reentrant locks would require a stack-like counting
    mechanism, which goes beyond the expressiveness of plain PNs;
    hence, it belongs to future work.
  } \\
  For example, if \J{lock} is a reentrant lock variable,
  the \Ji{error} location is reachable in \autoref{fig:soundness-examples-locking}'s Jimple program (left),
  where the same thread acquires \J{lock} twice in a row;
  in contrast, the second \RB{acquire} of \J{lock} in the corresponding \rb program (right)
  never executes (since \autoref{fig:rb-semantics}'s semantics does not allow it),
  and hence \RB{error} becomes effectively unreachable.

\item[Calls:] \tech's sound translation of call instructions depends on
  the capabilities of Soot's analysis. In particular,
  \Ji{invokedynamic} instructions are correctly translated only if
  Soot can retrieve the body of the invoked closure object.
  This is possible in simple cases such as
  \autoref{fig:motivating-example}'s example, where \Ji{invokedynamic} is used to execute
  the \J{Runnable} lambdas \J{lock_12} and \J{lock_21}; more complex
  instances of \Ji{invokedynamic} would become \RB{skip} in $R$,
  which introduces unsoundness in general. \\
  For example, consider an \Ji{invoke dl}, where \Ji{dl} resolves
  to a method that causes a deadlock when executed;
  if Soot cannot resolve \Ji{dl} statically,
  the call becomes simply \RB{skip} in \rb,
  and hence \tech would unsoundly conclude that there is no deadlock.

\item[Unsupported instructions:] Other bytecode instructions that
  are not listed in \autoref{fig:syntax-jimple} are currently
  unsupported by \tech.
  The translation replaces any unsupported instruction $I$
  with a \RB{skip}, which means that $R$ doesn't model $I$'s semantics.
  This may result in a loss of soundness or precision, depending on what execution paths the unsupported instruction
  does enable or block. \\
  For example, consider the Jimple code in \autoref{fig:soundness-examples-unsupported} (left):
  the \Ji{throw} instruction unconditionally jumps to location $\ell_2$, where a matching \Ji{catch} block is defined;
  therefore, location $\ell_1$ is unreachable in the original program, since it is just after the \Ji{throw}.
  Since \tech does not currently support exception-related instructions such as \Ji{throw},
  this instruction becomes simply a \RB{skip} in the \rb translation (right);
  thus, $\ell_1$ is reachable in the \rb program, whereas $\ell_2$ is unreachable.
  As a result of this translation, analyzing the \rb program in lieu of the Jimple program
  may be unsound (if the code at $\ell_2$ introduces an error, which would go undetected in the \rb program),
  imprecise (if the unreachable code at $\ell_1$ introduces an error, which would be spuriously detected in the \rb program),
  or possibly neither (if the code at $\ell_1$ and $\ell_2$ do not affect program correctness w.r.t.\ the properties
  \tech analyzes).

\item[Aliasing:] Two variables $v_1, v_2$ that may be \emph{aliased} in $J$ are lumped together
  into a single variable $v = \alpha(v_1) = \alpha(v_2)$ in $R$;
  this may introduce a loss of \emph{soundness} or, more commonly, \emph{precision}.\\
  \autoref{fig:soundness-examples-aliasing-unsound} shows an example where
  imprecise aliasing information introduces unsoundness:
  the two threads \Ji{t_1} and \Ji{t_2} acquire \Ji{lock_1} and \Ji{lock_2}
  in opposite order, which clearly may result in a deadlock.
  If the alias analysis erroneously determines that the two lock variables may be aliased,
  \tech models them as a single variable \RB{lock};
  thus, in the \rb encoding, a deadlock cannot occur because both threads try to acquire
  the same lock twice in a row (a case of reentrant locking, which is ignored by \rb's semantics
  as we discussed above).
  \\
  \autoref{fig:soundness-examples-aliasing-imprecise} shows an example where
  imprecise aliasing information introduces imprecision:
  thread \Ji{t_1} acquires \Ji{lock_1} and then \Ji{lock_2},
  whereas thread \Ji{t_2} acquires \Ji{lock_2} and then \Ji{lock_3}.
  If these are all distinct locks, no deadlock can occur.
  However, if the alias analysis erroneously determines that lock variables \Ji{lock_1} and \Ji{lock_3} may be aliased,
  \tech models them as a single variable \RB{lock_13};
  thus, in the \rb encoding, a deadlock may occur because both threads
  are in contention to acquire \RB{lock_13} and \RB{lock_2} 
  in opposite order.
\end{description}

\begin{figure}[!bt]
  \centering
  {\lstset{style=displayed-nln}\begin{subfigure}[t]{0.37\paperwidth}
    \begin{minipage}[t]{0.19\paperwidth}
\begin{lstlisting}[language=Jimple]
       if false error
       goto ok
error: $\ldots$
   ok: $\ldots$
\end{lstlisting}
    \end{minipage}
    \hfill
    \begin{minipage}[t]{0.17\paperwidth}
\begin{lstlisting}[language=RockBottom]
       jump $\ell$ error
    $\ell$: goto ok
error: $\ldots$
   ok: $\ldots$
\end{lstlisting}
    \end{minipage}
    \caption{A conditional in bytecode (left) whose \rb translation (right) is imprecise.}
    \label{fig:soundness-examples-conditionals}
  \end{subfigure}
\hfill
\begin{subfigure}[t]{0.42\paperwidth}
    \begin{minipage}[t]{0.22\paperwidth}
\begin{lstlisting}[language=Jimple]
      monitor_enter lock
      monitor_enter lock
error: $\ldots$
\end{lstlisting}
    \end{minipage}
    \hfill
    \begin{minipage}[t]{0.19\paperwidth}
\begin{lstlisting}[language=RockBottom]
       acquire lock
       acquire lock
error: // unreachable
\end{lstlisting}
    \end{minipage}
    \caption{A lock used reentrantly in bytecode (left)
      whose \rb translation (right) is unsound.}
    \label{fig:soundness-examples-locking}
  \end{subfigure}
\\[3mm]
\begin{subfigure}[t]{0.5\paperwidth}
    \begin{minipage}[t]{0.16\paperwidth}
\begin{lstlisting}[language=Jimple]
   throw e
$\ell_1$: $\ldots$

// catch(e) block
$\ell_2$: $\ldots$
\end{lstlisting}
    \end{minipage}
    \hfill
    \begin{minipage}[t]{0.16\paperwidth}
\begin{lstlisting}[language=RockBottom]
   skip
$\ell_1$: $\ldots$

// catch(e) block
$\ell_2$: $\ldots$
\end{lstlisting}
    \end{minipage}
    \caption{Translating an unsupported bytecode instruction such as \Ji{throw} (left) to \RB{skip} in \rb (right)
      may introduce unsoundness (if $\ell_2$ is an error location, unreachable in \rb but reachable in Jimple)
      or imprecision (if $\ell_1$ is an error location, reachable in \rb but unreachable in Jimple).}
    \label{fig:soundness-examples-unsupported}
  \end{subfigure}
\\[3mm]
\begin{subfigure}[t]{0.38\paperwidth}
    \begin{minipage}[t]{0.19\paperwidth}
\begin{lstlisting}[language=Jimple]
// thread t_1
monitor_enter lock_1
monitor_enter lock_2

// thread t_2
monitor_enter lock_2
monitor_enter lock_1
\end{lstlisting}
    \end{minipage}
    \hfill
    \begin{minipage}[t]{0.13\paperwidth}
\begin{lstlisting}[language=RockBottom]
// thread t_1
acquire lock
acquire lock

// thread t_2
acquire lock
acquire lock
\end{lstlisting}
    \end{minipage}
    \caption{If the alias analysis of the bytecode program on the left
      indicates that \J{lock_1} and \J{lock_2}
      may be aliases even though they are actually not,
      the \rb translation (right) would render them as a single \J{lock} variable. This is unsound
    because the \rb program would not deadlock even if the bytecode program it translates may deadlock.}
    \label{fig:soundness-examples-aliasing-unsound}
  \end{subfigure}
\hfill
\begin{subfigure}[t]{0.38\paperwidth}
    \begin{minipage}[t]{0.19\paperwidth}
\begin{lstlisting}[language=Jimple]
// thread t_1
monitor_enter lock_1
monitor_enter lock_2

// thread t_2
monitor_enter lock_2
monitor_enter lock_3
\end{lstlisting}
    \end{minipage}
    \hfill
    \begin{minipage}[t]{0.14\paperwidth}
\begin{lstlisting}[language=RockBottom]
// thread t_1
acquire lock_13
acquire lock_2

// thread t_2
acquire lock_2
acquire lock_13
\end{lstlisting}
    \end{minipage}
    \caption{If the alias analysis of the bytecode program on the left
      indicates that \J{lock_1} and \J{lock_3}
      may be aliases even though they are actually not,
      the \rb translation (right) would render them as a single \J{lock_13} variable. This is imprecise
    because the \rb program may deadlock even if the bytecode program it translates obviously does not deadlock.}
    \label{fig:soundness-examples-aliasing-imprecise}
  \end{subfigure}
}\caption{Examples of Jimple programs whose \rb translation by \tech may be unsound or imprecise.}
  \label{fig:soundness-examples-bytecode2rb}
\end{figure}

\section{Additional details for \autoref{sec:rb-to-pn}}

The PN encoding of commands in \autoref{fig:rb-to-pn} is as follows:
\begin{enumerate}[label=\emph{\alph*})]
\item Commands \RB{skip}, \RB{read}, and \RB{write}
  transfer the token from place $(p, t, \ell, r)$ to the next location's
  place $(p, t, \ell + 1, r)$.
\item Command $\RB{goto}\ \ell'$
  transfers the token from the place $(p, t, \ell, r)$ to the unconditional jump's 
  target's place $(p, t, \ell', r)$.
\item Similarly, command $\RB{jump}\:\ell_1\:\ell_2$
  nondeterministically transfers the token to either place $(p, t, \ell_1, r)$
  or place $(p, t, \ell_2, r)$.
\item Command $\RB{fork}\:p'\:t'$ puts one token into the next location's place,
  and one token into the forked thread's entry point's place $(p', t', \RB{entry}_{p'}, r)$.
  This way, the forked thread's computation can proceed in parallel to the forking thread's.
\item Conversely, $\RB{join}\ t'$ can fire only when
  the exit transition $(p', t', \RB{exit}_{p'}, r)$ of thread $t'$
  fires (where $p'$ is the procedure that thread $t'$ is running).
\item Command $\RB{acquire}\:v$'s transition can fire only if \emph{all}
  places $(t', v)$ are marked, denoting that no thread $t'$ holds a lock on $v$.
  When it fires, it puts back a token only in place $(t, v)$
  to indicate that $t$ holds a lock on $v$.
\item Command $\RB{release}\:v$'s transition can fire only if place $(t, v)$
  is marked, denoting that thread $t$ holds a lock on $v$.
  When it fires, it puts back a token in all places $(t', v)$,
  thus allowing other threads to acquire a lock on $v$.
\item When the transition $(p, t, \ell, r)$ of command $\RB{call}\:p'$
  fires, it puts a token in the callee's entry place $(p', t, \RB{entry}_{p'}, \ell)$.
  Then, execution continues in the subnet corresponding procedure $p'$ called in thread $t$
  at call site $\ell$.
  Then, when transition $(p', t, \RB{exit}_{p'}, \ell)$ fires,
  signaling that the subnet's execution has terminated,
  a token first goes into return place $(p', t, \rets{\ell})$;
  then one of the transitions $(p, t, \rets{\ell}, r)$, for all possible callers of $p'$,
  nondeterministically fires.
  \autoref{fig:rb2pn:call} pictures two callers $r_1 \neq r_2$,
  where either of the return transitions $(p, t, \rets{\ell}, r_k)$, $k = 1, 2$,
  may fire when the subnet terminates execution. \\
  This nondeterministic encoding of returns is the only aspect
  of \tech's encoding of \rb programs that loses precision:
  a PN cannot store an unbounded stack of return locations,\footnote{
    PNs with inhibitor arcs would be able to simulate this
    without loss of precision; however, their reachability problem becomes undecidable~\cite{FMMR-TimeBook-12}.
  }
  and hence \tech overapproximates it with a nondeterministic choice.
\end{enumerate}

\begin{proof}[Proof outline of \autoref{lemma:pn-correct}]
  Let $S_0 \leadsto S_1 \leadsto \cdots$
  be a sequence of states
  in $P$'s semantics.
Each $S_k$ maps to a marking $m_k$ of $N$ as follows.
  For every $(t, p, \ell, K, R, \tau) \in S_k$,
  let $\rho = \top$ if $R = \emptyset$,
  and $\rho = r$ if $R = R' + [r + 1]$.
  Then, the following places are marked in $N$:
  \begin{enumerate*}
  \item $(p, t, \ell, \rho)$ if $\tau = \running$;
  \item $(p, t, \joins{\ell}, \rho)$ if $\tau = \terminated$;
  \item $(t, k)$ for every $k \in K$.
  \end{enumerate*}
  Furthermore, for every $k'$ such that $k' \not\in K'$ for every
  tuple $(t', p', \ell', K', R', \tau) \in S$, all places $(t', k')$ are also marked.
In particular, this mapping translates $P$'s initial state
  $\{(t_0, \RB{main}, \RB{entry}_{\RB{main}}, \emptyset, \emptyset, \running)\}$
  into $N$'s initial marking.

  Then, one can show by induction that if $S_k \leadsto S_{k+1}$
  then $m_k \vdash^{1,2} m_{k+1}$,
  where $\vdash^{1,2}$ denotes one or two steps in $N$'s semantics.
  Precisely, the only scenario when a step in $P$'s evaluation
  corresponds to \emph{two} consecutive markings in $N$'s semantics
  is when $s = (t, p, \RB{exit}_p, L, R + [\ell], \running) \in S_k \leadsto S_k \setminus \{s\} \cup \{(t, p', \ell, L, R, \running)\}$;
  in this case, $m_k \vdash m_k' \vdash m_{k+1}$
  where $m_k'$ is the marking reached after transition $(p, t, \RB{exit}_{p}, \ell)$
  fires, moving a token from place $(p, t, \RB{exit}_{p}, \ell)$
  to $(p, t, \rets{\ell})$,
  and $m_{k+1}$ corresponds to the nondeterministically chosen return location
  (as shown in \autoref{fig:rb2pn:call}).
\end{proof}

\section{Additional details for \autoref{sec:encoding-props}}

\tech takes care of expressing these properties
for the PN analyzer, so
that the user does not have to directly interact with the latter.
For this work, we focus on three widely useful concurrency properties:
termination, deadlock, and livelock.
In future work, we will extend this approach to support
other concurrency properties,
such as data races and atomicity---even though
one can already analyze these properties by directly expressing them
in the language of the PN analyzer.

\begin{table}[!bt]
  \centering
  \setlength{\tabcolsep}{1pt}
\begin{tabular}{clcl}
    \toprule
    && \multicolumn{1}{c}{\textsc{monitor}} & \multicolumn{1}{c}{\textsc{ctl temporal logic formula}} \\
    \midrule
    $\phi_T$ & termination & \multirow{3}{*}{\begin{tikzpicture}[baseline=(current bounding box.center),node distance=8mm and 3mm,-latex]
                               \node (start) {};
                     \pntrans{t-entry}{$p$}[$t$][$\RB{entry}\ r$][below=3mm of start]
                     \pntrans{t-exit}{$p$}[$t$][$\RB{exit}\ r$][below=of t-entry]

                     \begin{scope}[every node/.append style={fill=pncol,draw=pncol}]
                     \pnplace{p-spawn}{$p$}[$t\ r$][\textcolor{pncol}{\spawnMark}][right=of t-entry]
                     \pnplace{p-term}{$p$}[$t\ r$][\textcolor{pncol}{\termMark}][below=of p-spawn]
                     \coordinate (mid) at ($(p-spawn)!0.5!(p-term)$);
                     \pntrans{t-done}{$p$}[$t\ r$][\textcolor{pncol}{\terminated}][right=of mid]
                     \end{scope}

                     \begin{scope}[draw=pncol,thick]
                     \draw (t-entry) to (p-spawn);
                     \draw (t-exit) to (p-term);
                     \draw (p-spawn) to (t-done);
                     \draw (p-term) to (t-done);
                   \end{scope}
                 \end{tikzpicture}
                   }
    & $\Atl\Ftl[\pndeadlock \land \forall p\colon P, t\colon T, r\colon L, \tau\colon \{\spawnMark, \termMark\} \cdot m(p, t, r, \tau) = 0]$
    \\
    $\phi_D$ & deadlock & \rule{0pt}{17mm}
   & $\Etl\Ftl\,\Etl\Gtl[\pndeadlock \land
       \exists p\colon P, t\colon T, r\colon L, \tau\colon \{\spawnMark, \termMark\} \cdot m(p, t, r, \tau) \neq 0]$
    \\
    $\phi_L$ & livelock & 
   & $\Etl\Ftl\,\Etl\Gtl[\neg \pndeadlock \land
       \exists p\colon P, t\colon T, r\colon L, \tau\colon \{\spawnMark, \termMark\} \cdot m(p, t, r, \tau) \neq 0]$
    \\
    \bottomrule
  \end{tabular}
\caption{How \tech checks for various concurrency properties.}
  \label{tab:properties}
\end{table}

In order to analyze a property $\phi$ on a PN $N$,
\tech first \emph{extends} $N$ with additional places and transitions
that act as \emph{monitors} of $N$'s state components
that are useful to check for the properties.
While using monitors is not strictly needed,
it helps simplify the temporal logic formula that expresses the properties,
which in turn results in performance benefits.
\autoref{tab:properties} shows the monitors and temporal logic formulas
built by \tech to verify properties \emph{termination} $\phi_T$ (which holds if the program always terminates),
\emph{deadlock} $\phi_D$ (which holds if the program may deadlock) and
\emph{livelock} $\phi_L$ (which holds if the program may livelock).
These three properties use the same monitor:
for each subnet corresponding to thread $t$ executing procedure $p$ with return location $r$, the monitor adds two places $(p, t, r, \spawnMark)$ and
$(p, t, r, \termMark)$, and a transition $(p, t, r, \terminated)$.
As soon as a subnet's \RB{entry} (resp.~\RB{exit}) transition fires,
place $(p, t, r, \spawnMark)$ (resp.~$(p, t, r, \termMark)$)
gets a token; when both places are marked, transition $(p, t, r, \terminated)$
fires and empties them.
Therefore, $(p, t, r, \spawnMark)$ is marked iff the subnet is executing.

With this monitor, \autoref{tab:properties}'s CTL formulas
express the three properties of termination, deadlock, and livelock.
The formulas use predicate \pndeadlock, which is built-in most PN analyzers
and denotes a PN deadlock: a situation where all transitions in the PN
are permanently disabled.
In \tech's encoding $N$, a PN deadlock does not necessarily
correspond to a deadlock of program $J$:
if all threads have completed execution normally (or have never been started),
a PN deadlock simply denotes normal termination;
but if some threads have not completed execution,
a PN deadlock corresponds to a program deadlock.
\tech correctly distinguishes between deadlocks and termination
also in cases where some threads are forked but not joined by
expressing properties on the monitor's places as follows:
\begin{enumerate}
\item \emph{Termination}: the program eventually terminates
  iff all the threads that have started eventually finish execution.
  Thus, $\phi_T$ checks that, along 
  all execution paths in the future (CTL operator $\Atl\Ftl$),
  the PN deadlocks (predicate \pndeadlock, which indicates that no transition can fire)
  and all \spawnMark\ and \termMark\ places are empty (which indicates that no thread is executing).
\item \emph{Deadlock}: the program is stuck and cannot reach proper termination or make progress.
  Thus, $\phi_D$ checks that,
  along all execution paths from some point on in the future
  (CTL operator $\Etl\Ftl\,\Etl\Gtl$),
  the PN deadlocks (predicate \pndeadlock, which indicates that no transition can fire) and at least one
  \spawnMark\ and \termMark\ place remains not empty
  (which indicates that some thread has started but cannot complete execution).
  The non-empty places
  correspond to the deadlocked threads.
\item \emph{Livelock}: the program as a whole can continue execution,
  but one or more threads cannot terminate or make progress.
  Thus, $\phi_L$ checks that, 
  along all execution path from some point on in the future
  (CTL operator $\Etl\Ftl\,\Etl\Gtl$),
  the PN does \emph{not} deadlock ($\neg \pndeadlock$, since the program as a whole still runs)
  but at least one
  \spawnMark\ and \termMark\ place remains not empty indefinitely
  (which indicates that some thread has started but cannot complete execution). The non-empty places
  correspond to the threads that can't reach termination.
\end{enumerate}

\begin{expara}
  \autoref{fig:ex:pn}'s parts in \colorbox{pncol}{\textcolor{white}{purple}}
  monitor the termination of each thread as in \autoref{tab:properties}.
  Since the main thread does not join the two spawned threads,
  this is a situation where a deadlock of the PN (i.e., $\Etl\Ftl\,\Etl\Gtl[\pndeadlock]$ holds)
  does not imply a deadlock of the original program.
  This is where \tech's monitor is used to distinguish the two scenarios:
  if all procedures terminate normally, \autoref{tab:properties}'s formula
  $\phi_D$ will never hold, since all $\spawnMark, \termMark$ places will
  eventually be empty.
  In contrast, when the PN reaches the state where the \colorbox{red!20}{light red}
  places are marked, places
  $(\RB{t_12}, \spawnMark)$ and
  $(\RB{t_21}, \spawnMark)$ will remain marked too;
  thus, the state exposes a genuine deadlock in \autoref{fig:ex:java}'s original program,
  which satisfies formula $\phi_D$.
\end{expara}

\subsection{Implementation Details and Limitations}
\label{sec:limitations}

We implemented the \tech technique in a command-line tool called \tool.
\tool inputs a bytecode program $J$,
uses Soot to analyze Jimple code and its control-flow,
and produces a PN $P = \encoding(\translation(J))$ as described in previous sections.

As we discussed previously in this section,
\tech's output $P$ is a \emph{sound} model of
the execution order of instructions in $J$ provided
the following conditions are met:
\begin{enumerate}
\item Soot's alias analysis of lock variables in $J$
  is sufficiently accurate (as explained in \autoref{sec:jimple-to-rb});
\item $J$ does not use unsupported bytecode features
  (mainly, exceptions and \Ji{invokedynamic} calls that Soot cannot resolve);
\item $J$ does not use any lock reentrantly
  (i.e., a thread acquires a lock on a variable it's already locking).
\end{enumerate}
Under these conditions, if a temporal logic property $\phi$ holds on $P$,
then it also holds on $J$.

Conversely, \tech's output $P$'s \emph{precision}
as a model of $J$'s executions depends on several factors:
\begin{enumerate} 
\item The accuracy of the points-to alias analysis also affects precision.
  Currently, \tool uses Soot's Spark whole-program analysis,
  which induces a significant loss of precision with features such as array indexing.
\item Path-sensitive information is ignored in $P$,
  which means that all control-flow paths in $J$ are feasible in $P$.
\item Context-sensitive information is overapproximated by $P$,
  so a call may nondeterministically return to any of its
  possible call sites in $J$.
\item Finally, below we outline how \tech's modeling of loops and recursion
  affects its precision.
\end{enumerate}
The experiments in \autoref{sec:experiments}
assess the practical impact of \tech's current limitations.

\begin{figure}[!tbh]
  \begin{subfigure}[c]{0.35\linewidth}
\begin{lstlisting}[language=JavaRecent,style=displayed,numbers=none]
void main() {
  do {
     Thread t = new Thread(() -> { $p$ });
     t.start();
     if ($\ldots$) t.join();
  } while ($\ldots$);
}
\end{lstlisting}
  \end{subfigure}
  \hspace{3mm}
  \begin{subfigure}[c]{0.25\linewidth}
    \scriptsize
    \begin{tabular}{rl}
      & \RB{def main} \\
      $\RB{entry}_{\RB{main}}\colon$ & \RB{begin} \\
      $\ell_0\colon$ & \RB{fork $\ p$ t} \\
      $\ell_1\colon$ & \RB{jump}$\ \ell_2\ \ell_3$ \\
      $\ell_2\colon$ & \RB{join t} \\
      $\ell_3\colon$ & \RB{jump}$\ \ell_0, \RB{exit}_{\RB{main}}$ \\
      $\RB{exit}_{\RB{main}}\colon$ & \RB{end}
    \end{tabular}
  \end{subfigure}
  \\[2mm]
  \begin{subfigure}[c]{0.82\textwidth}
    \centering
    \scriptsize
    \begin{tikzpicture}[-latex,node distance=7mm and 7mm,marked/.style={draw=blue!40,very thick}]
      \pnplace{p-m-e}{\RB{m}}[$t_0$][$\RB{entry}_{\RB{m}}$][tokens=1]
      \pntrans{t-m-e}{\RB{m}}[$t_0$][$\RB{entry}_{\RB{m}}$][right=of p-m-e]
      \pnplace{p-m-0}{\RB{m}}[$t_0$][$\ell_0$][right=of t-m-e]
      \pntrans{t-m-0}{\RB{m}}[$t_0$][$\ell_0$][right=of p-m-0]
      \pnplace{p-m-1}{\RB{m}}[$t_0$][$\ell_1$][right=of t-m-0]
      \pntrans{t-m-1-2}{\RB{m}}[$t_0$][$\ell_1 \to \ell_2$][right=of p-m-1]
      \pntrans{t-m-1-3}{\RB{m}}[$t_0$][$\ell_1 \to \ell_3$][above=of t-m-1-2]
      \pnplace{p-m-2}{\RB{m}}[$t_0$][$\ell_2$][right=of t-m-1-2]
      \pntrans{t-m-2}{\RB{m}}[$t_0$][$\ell_2$][right=of p-m-2]
      \pnplace{p-m-j2}{\RB{m}}[$t_0$][$\joins{\ell_2}$][below=of t-m-2]
      \pnplace{p-m-3}{\RB{m}}[$t_0$][$\ell_3$][right=of t-m-2]
      \pntrans{t-m-3-x}{\RB{m}}[$t_0$][$\ell_3 \to \RB{exit}_{\RB{m}}$][right=of p-m-3]
      \pntrans{t-m-3-0}{\RB{m}}[$t_0$][$\ell_3 \to \ell_0$][above=12mm of p-m-3]
      \pnplace{p-m-x}{\RB{m}}[$t_0$][$\RB{exit}_{\RB{m}}$][right=of t-m-3-x]
      \pntrans{t-m-x}{\RB{m}}[$t_0$][$\RB{exit}_{\RB{m}}$][right=of p-m-x]
\pnplace{p-p-e}{$p$}[\RB{t}][$\RB{entry}_{p}$][below=12mm of t-m-0]\
      \node (p-body) [right=of p-p-e] {$\cdots$};
      \pntrans{t-p-x}{$p$}[\RB{t}][$\RB{exit}_{p}$][right=of p-body]

      \foreach \x/\y in {p-m-e/t-m-e,t-m-e/p-m-0,p-m-0/t-m-0,t-m-0/p-m-1,p-m-1/t-m-1-2,p-m-1/t-m-1-3,t-m-1-2/p-m-2,p-m-2/t-m-2,t-m-2/p-m-3,p-m-3/t-m-3-x,p-m-3/t-m-3-0,t-m-3-x/p-m-x,p-m-x/t-m-x,t-m-1-3/p-m-3,t-m-0/p-p-e,t-p-x/p-m-j2,p-m-j2/t-m-2} {\draw (\x) to (\y);}

      \draw (t-m-3-0) -| (p-m-0);
      \draw[dotted] (p-p-e) to (p-body);
      \draw[dotted] (p-body) to (t-p-x);
    \end{tikzpicture}
  \end{subfigure}
    \caption{Petri net encoding of an \rb program that models a Java program that spawns threads in a loop. For readability, procedure \J{main} is abbreviated as \J{m} in the Petri net.}
    \label{fig:ex:loop-fork}
\end{figure}

\begin{figure}
  \begin{subfigure}[c]{0.35\linewidth}
\begin{lstlisting}[language=JavaRecent,style=displayed,numbers=none]
var lock = new Lock();

static void main() {
  rec();
  lock.unlock();
}

static void rec() {
  lock.lock();
  if ($\ldots$) { lock.unlock(); rec(); }
}
\end{lstlisting}
  \end{subfigure}
  \hspace{5mm}
  \begin{subfigure}[c]{0.25\linewidth}
    \scriptsize
    \begin{tabular}{rl}
      & \RB{def main} \\
      $\RB{entry}_{\RB{main}}\colon$ & \RB{begin} \\
      $\ell_0\colon$ & \RB{call rec} \\
      $\ell_1\colon$ & \RB{release lock} \\
      $\RB{exit}_{\RB{main}}\colon$ & \RB{end} \\
      \\
      & \RB{def rec} \\
      $\RB{entry}_{\RB{rec}}\colon$ & \RB{begin} \\
      $r_0\colon$ & \RB{acquire lock} \\
      $r_1\colon$ & \RB{jump}$\ r_2\ \ \RB{exit}_{\RB{rec}}$ \\
      $r_2\colon$ & \RB{release lock} \\
      $r_3\colon$ & \RB{call rec} \\
      $\RB{exit}_{\RB{rec}}\colon$ & \RB{end} \\
    \end{tabular}
  \end{subfigure}
  \\[3mm]
  \begin{subfigure}[c]{0.95\textwidth}
    \centering
    \scriptsize
    \begin{tikzpicture}[-latex,node distance=10mm and 6mm]
      \pnplace{p-m-e}{\RB{m}}[][$\RB{en}_{\RB{m}}\top$][tokens=1]
      \pntrans{t-m-e}{\RB{m}}[][$\RB{en}_{\RB{m}}\top$][right=of p-m-e]
      \pnplace{p-m-0}{\RB{m}}[][$\ell_0 \top$][right=of t-m-e]
      \pntrans{t-m-0}{\RB{m}}[][$\ell_0 \top$][right=of p-m-0]
\pnplace{p-r-e}{\RB{r}}[][$\RB{en}_{\RB{r}}\:\ell_0$][below=22mm of t-m-0]
      \pntrans{t-r-e}{\RB{r}}[][$\RB{en}_{\RB{r}}\:\ell_0$][right=of p-r-e]
      \pnplace{p-r-0}{\RB{r}}[][$r_0\:\ell_0$][right=of t-r-e]
      \pntrans{t-r-0}{\RB{r}}[][$r_0\:\ell_0$][right=of p-r-0,acq]
      \pnplace{p-r-1}{\RB{r}}[][$r_1\:\ell_0$][right=of t-r-0]
      \pntrans{t-r-1-2}{\RB{r}}[][$r_1 \!\!\to\!\!r_2\:\ell_0$][above=of t-r-0]
      \pnplace{p-r-2}{\RB{r}}[][$r_2\:\ell_0$][right=of t-r-1-2]
      \pntrans{t-r-2}{\RB{r}}[][$r_2\:\ell_0$][right=of p-r-2,rel]
      \pnplace{p-r-3}{\RB{r}}[][$r_3\:\ell_0$][right=of t-r-2]
      \pntrans{t-r-3}{\RB{r}}[][$r_3\:\ell_0$][right=of p-r-3]
      \pnplace{p-r-x}{\RB{r}}[][$\RB{ex}_{\RB{r}}\:\ell_0$][right=of t-r-3]
      \pntrans{t-r-x}{\RB{r}}[][$\RB{ex}_{\RB{r}}\:\ell_0$][right=of p-r-x]
      \pntrans{t-r-1-x}{\RB{r}}[][$r_1 \!\!\to\!\!\RB{ex}_{\RB{r}}\:\ell_0$][right=of p-r-1]
      \pnplace{p-r-ret}{\RB{r}}[][$\rets{\ell_0}$][above=of t-r-x]
      \pntrans{t-r-ret}{\RB{r}}[][$\rets{\ell_0}\:\top$][left=of p-r-ret]
\pnplace{p-m-1}{\RB{m}}[][$\ell_1 \top$][above=of t-r-2]
      \pntrans{t-m-1}{\RB{m}}[][$\ell_1 \top$][left=of p-m-1,rel]
      \pnplace{p-m-x}{\RB{m}}[][$\RB{ex}_{\RB{m}}\top$][left=of t-m-1]
      \pntrans{t-m-x}{\RB{m}}[][$\RB{ex}_{\RB{m}}\top$][left=of p-m-x]
\pnplace{p-rr-e}{\RB{r}}[][$\RB{en}_{\RB{r}}\:r_3$][below=22mm of t-r-3]
      \pntrans{t-rr-e}{\RB{r}}[][$\RB{en}_{\RB{r}}\:r_3$][left=of p-rr-e]
      \pnplace{p-rr-0}{\RB{r}}[][$r_0\:r_3$][left=of t-rr-e]
      \pntrans{t-rr-0}{\RB{r}}[][$r_0\:r_3$][left=of p-rr-0,acq]
      \pnplace{p-rr-1}{\RB{r}}[][$r_1\:r_3$][left=of t-rr-0]
      \pntrans{t-rr-1-2}{\RB{r}}[][$r_1 \!\!\to\!\!r_2\:r_3$][below=of p-rr-1]
      \pnplace{p-rr-2}{\RB{r}}[][$r_2\:r_3$][left=of t-rr-1-2]
      \pntrans{t-rr-2}{\RB{r}}[][$r_2\:r_3$][left=of p-rr-2,rel]
      \pnplace{p-rr-3}{\RB{r}}[][$r_3\:r_3$][left=of t-rr-2]
      \pntrans{t-rr-3}{\RB{r}}[][$r_3\:r_3$][left=of p-rr-3]
      \pnplace{p-rr-x}{\RB{r}}[][$\RB{ex}_{\RB{r}}\:r_3$][right=of p-rr-e]
      \pntrans{t-rr-x}{\RB{r}}[][$\RB{ex}_{\RB{r}}\:r_3$][right=of p-rr-x]
      \pntrans{t-rr-1-x}{\RB{r}}[][$r_1 \!\!\to\!\!\RB{ex}_{\RB{r}}\:r_3$][below=of p-rr-0]
\pnplace{p-rr-ret}{\RB{r}}[][$\rets{r_3}$][right=of t-rr-x]
      \pntrans{t-rr-ret-l0}{\RB{r}}[][$\rets{r_3}\:\ell_0$][above=of t-rr-x]
      \pntrans{t-rr-ret-r3}{\RB{r}}[][$\rets{r_3}\:r_3$][below=of t-rr-x]

      \foreach \x/\y in {p-m-e/t-m-e,t-m-e/p-m-0,p-m-0/t-m-0,t-m-0/p-r-e,p-m-1/t-m-1,t-m-1/p-m-x,p-m-x/t-m-x,p-r-e/t-r-e,t-r-e/p-r-0,p-r-0/t-r-0,t-r-0/p-r-1,p-r-1/t-r-1-2,p-r-1/t-r-1-x,t-r-1-2/p-r-2,p-r-2/t-r-2,t-r-2/p-r-3,p-r-3/t-r-3,p-r-x/t-r-x,
        t-r-3/p-rr-e,p-rr-e/t-rr-e,t-r-x/p-r-ret,p-r-ret/t-r-ret,t-r-ret/p-m-1,t-rr-e/p-rr-0,p-rr-0/t-rr-0,t-rr-0/p-rr-1,p-rr-1/t-rr-1-2,t-rr-1-2/p-rr-2,p-rr-2/t-rr-2,t-rr-2/p-rr-3,p-rr-3/t-rr-3,p-rr-1/t-rr-1-x,p-rr-x/t-rr-x,t-rr-x/p-rr-ret,p-rr-ret/t-rr-ret-l0,p-rr-ret/t-rr-ret-r3,t-rr-ret-l0/p-r-x,t-rr-ret-r3/p-rr-x} {\draw (\x) to (\y);}

      \draw (t-r-1-x) -| (p-r-x);
      \draw (t-rr-1-x) -| (p-rr-x);

      \coordinate[below=4mm of t-rr-3] (below-t-rr-3);
      \draw (t-rr-3) -- (below-t-rr-3) -| (p-rr-e);
    \end{tikzpicture}
  \end{subfigure}
  \caption{Petri net encoding of an \rb program that models a Java program that recursively acquires and releases a lock. For readability, procedures \J{main} and \J{rec} are abbreviated as \J{m} and \J{r} in the Petri net; labels \RB{entry} and \RB{exit} are abbreviated as \RB{en} and \RB{ex}; and the places for lock variable \RB{lock} are not shown explicitly, but the transitions
  corresponding to \RB{acquire} and \RB{release} are highlighted \stylenode{acq} and \stylenode{rel}.}
  \label{fig:ex:recursion}
\end{figure}

\subsubsection{Threads and Loops.}
\autoref{fig:ex:loop-fork}
demonstrates how unbounded loops are handled
by \tech, and how they affect precision.
The Java program shown there
starts a certain number of threads
in a loop and joins some of them.
As long as the code $p$ executed by the spawned threads
does not introduce a circular wait dependency with the
main thread, the program does not deadlock.
Let us see how this behavior is captured accurately by the \rb model,
and in turn by the PN encoding, shown in the same figure.

First, notice that the set $T$ of thread identifiers in the \rb program
is $\{t_0, \RB{t}\}$;
in general, since threads are distinguished by their program identifiers up to aliasing,
$T$ is always finite in an \rb program.
For each started thread, a new tuple $(p, \RB{t}, \RB{entry}_{p}, \emptyset, \emptyset, \running)$
is added to the state $S$;
correspondingly, a token is added to place $(p, \RB{t}, \RB{entry}_{p})$
in the PN,
modeling the asynchronous execution of several threads.

In general, the PN and the \rb program have infinitely many executions,
whereas the Java program will probably only execute the loop
a finite number of times.
Another source of overapproximation is the conditional
in the loop body, which may be deterministic in the Java
program but is nondeterministic in the \rb and PN models.
Nevertheless, the analysis of deadlock behavior is
still precise on the PN model, since whether
the Java program deadlocks does not depend on how many times
the loop or conditional are executed.
This example demonstrates how the counting capabilities of PNs
are sufficient to capture recurring thread spawning patterns
while retaining precision in the analysis of certain concurrency
properties.

\subsubsection{Recursion.}
\autoref{fig:ex:recursion} demonstrates how recursion in handled in \tool.
The Java program shown there starts executing a recursive method \RB{rec};
with each recursive call, a lock is acquired before deciding whether to
continue with another recursive call or return to the caller.
The \rb program in \autoref{fig:ex:recursion} has the same behavior
as the Java program, except for the nondeterministic control flow;
concretely, this means that the \rb program has infinitely many possible executions,
one for each possible maximum recursion depth $n \geq 0$.

\autoref{fig:ex:recursion} also shows (with minor simplifications discussed in the caption)
the PN encoding the \rb program built by \tech.
The top row of nodes corresponds to the procedure \RB{main}.
The rest of the PN consists of two structurally isomorphic subnets, each encoding procedure \RB{rec};
precisely, there is one subnet for each invocation site of \RB{rec}:
the subnet in the middle of the picture corresponds to \RB{call rec} at $\ell_0$,
while the bottom subnet corresponds to \RB{call rec} at $r_3$.
When execution in the latter subnet terminates---signaled by a token in place $(\RB{r}, \rets{r_3})$---the call nondeterministically returns to either of the two call sites.

This nondeterminism overapproximates the behavior of the \rb program
(and thus the Java program)
since it includes computations where $n > 0$ nested recursive calls return abruptly to \RB{main}.
This may reduce precision, since the PN
includes executions that are infeasible in the original program.
As the experiments in \autoref{sec:experiments} demonstrate,
this limitation in principle
does not always impact the practical capabilities of \tool:
as long as the additional executions introduced by the overapproximation
do not generate spurious violations of the concurrency properties of interest,
the loss of precision is immaterial.

\fi

\end{document}